\documentclass[11pt]{article}

\usepackage[margin=1in]{geometry}
\usepackage[T1]{fontenc}
\usepackage[latin9]{inputenc}
\usepackage{graphicx,verbatim}
\usepackage{amsmath}
\usepackage{amssymb}
\usepackage[ruled,vlined]{algorithm2e}

\allowdisplaybreaks
\graphicspath{{figures/} }

\newcommand{\bA}{{\boldsymbol A}}
\newcommand{\bB}{{\boldsymbol B}}
\newcommand{\bC}{{\boldsymbol C}}
\newcommand{\bD}{{\boldsymbol D}}

\newcommand{\bG}{{\boldsymbol G}}
\newcommand{\bH}{{\boldsymbol H}}
\newcommand{\bI}{{\boldsymbol I}}

\newcommand{\bK}{{\boldsymbol K}}
\newcommand{\bM}{{\boldsymbol M}}

\newcommand{\bP}{{\boldsymbol P}}
\newcommand{\bQ}{{\boldsymbol Q}}

\newcommand{\bS}{{\boldsymbol S}}
\newcommand{\bT}{{\boldsymbol T}}
\newcommand{\bU}{{\boldsymbol U}}
\newcommand{\bV}{{\boldsymbol V}}
\newcommand{\bW}{{\boldsymbol W}}
\newcommand{\bX}{{\boldsymbol X}}
\newcommand{\bY}{{\boldsymbol Y}}
\newcommand{\bZ}{{\boldsymbol Z}}

\newcommand{\by}{{\boldsymbol y}}

\newcommand{\cS}{{\mathcal S}}
\newcommand{\cE}{{\mathcal E}}

\newcommand{\argmin}{\mathrm{argmin}}

\usepackage{cite}\usepackage{amsthm}\usepackage{dsfont}\usepackage{array}\usepackage{mathrsfs}\usepackage{comment}

\allowdisplaybreaks

\usepackage{hyperref}
\usepackage{breakurl}

\usepackage{color}
\usepackage[normalem]{ulem}

\begin{document}

\theoremstyle{plain}\newtheorem{lemma}{\textbf{Lemma}}\newtheorem{theorem}{\textbf{Theorem}}\newtheorem{corollary}{\textbf{Corollary}}\newtheorem{assumption}{\textbf{Assumption}}\newtheorem{example}{\textbf{Example}}\newtheorem{definition}{\textbf{Definition}}\newtheorem{proposition}{\textbf{Proposition}}

\theoremstyle{definition}

\theoremstyle{remark}\newtheorem{remark}{\textbf{Remark}}

\title{Nonconvex Low-Rank Matrix Recovery with Arbitrary Outliers\\
via Median-Truncated Gradient Descent}

\author{Yuanxin Li$^{\dag}$, Yuejie Chi$^{\dag}$, Huishuai Zhang$^{\ddag}$, Yingbin Liang$^{\dag}$ \\
$^{\dag}$Department of ECE, The Ohio State University, Columbus, OH 43210\\
$^{\ddag}$Department of EECS, Syracuse University, Syracuse, NY 13244\\
Emails:  \{li.3822, chi.97, liang.889\}@osu.edu,  \{hzhan23\}@syr.edu}

\date{\today}

\maketitle

\begin{abstract}
Recent work has demonstrated the effectiveness of gradient descent for directly recovering the factors of low-rank matrices from random linear measurements in a globally convergent manner when initialized properly. However, the performance of existing algorithms is highly sensitive in the presence of outliers that may take arbitrary values. In this paper, we propose a truncated gradient descent algorithm to improve the robustness against outliers, where the truncation is performed to rule out the contributions of samples that deviate significantly from the {\em sample median} of measurement residuals adaptively in each iteration. We demonstrate that, when initialized in a basin of attraction close to the ground truth, the proposed algorithm converges to the ground truth at a linear rate for the Gaussian measurement model with a near-optimal number of measurements, even when a constant fraction of the measurements are arbitrarily corrupted. In addition, we propose a new truncated spectral method that ensures an initialization in the basin of attraction at slightly higher requirements.  We finally provide numerical experiments to validate the superior performance of the proposed approach.

\end{abstract}


\section{Introduction}

Low-rank matrix recovery is a problem of great interest in applications such as collaborative filtering, signal processing, and computer vision. A considerable amount of work has been done on low-rank matrix recovery in recent years, where it is shown that low-rank matrices can be recovered accurately and efficiently from much fewer observations than their ambient dimensions \cite{recht2010guaranteed,gross2011recovering, negahban2011estimation, candes2012exact, chen2014robust,chen2015exact}. An extensive overview on low-rank matrix recovery can be found in \cite{davenport2016overview}. In particular, convex relaxation is a popular strategy which replaces the low-rank constraint by a convex surrogate, such as nuclear norm minimization \cite{recht2010guaranteed, jain2010guaranteed, candes2011tight,chen2014robust,chen2015exact}. However, despite statistical (near-)optimality, the computational and memory costs of this approach are prohibitive for high-dimensional problems.  

In practice, a widely used alternative, pioneered by Burer and Monteiro \cite{burer2003nonlinear}, is to directly estimate the factors $\bX\in\mathbb{R}^{n_1\times r}$ and $\bY\in\mathbb{R}^{n_2\times r}$, of a low-rank matrix $\bM = \bX\bY^T\in\mathbb{R}^{n_1\times n_2}$ if its rank $r$ is approximately known or can be upper bounded. Since the factors have a much lower-dimensional representation, this approach admits more computationally and memory efficient algorithms. Due to the bilinear constraint induced by factorization, this typically leads to a nonconvex loss function that may be difficult to optimize globally. Interestingly, a series of recent work has demonstrated that, starting from a careful initialization, simple algorithms such as gradient descent \cite{zheng2015convergent, tu2016low, zhao2015nonconvex, chen2015fast, park2016provable} and alternating minimization \cite{jain2013low,hardt2014understanding} enjoy global convergence guarantees under near-optimal sample complexity. Some of these algorithms also converge at a linear rate, making them extremely appealing computationally. On the other hand, the global geometry of nonconvex low-rank matrix estimation has been investigated in \cite{bhojanapalli2016global,ge2016matrix,li2016nonconvex,li2016symmetry}, and it is proven that no spurious local optima, except strict saddle points, exist under suitable coherence conditions and sufficiently large sample size. This implies global convergence from random initialization, provided the algorithm of choice can escape saddle points \cite{ge2015escaping, lee2016gradient, jin2017escape}.

In real-world applications, it is quite typical that measurements may suffer from outliers that need to be addressed carefully. In this paper, we focus on low-rank matrix recovery from random linear measurements in the presence of arbitrary outliers. Specifically, the sensing matrices are generated with i.i.d. standard Gaussian entries. Moreover, we assume that a small number of measurements are corrupted by outliers, possibly in an adversarial fashion with arbitrary amplitudes. This setting generalizes the outlier-free models studied in \cite{recht2010guaranteed,candes2011tight,zheng2015convergent,tu2016low}, where convex and nonconvex approaches have been developed to accurately recover the low-rank matrix. Unfortunately, the vanilla gradient descent algorithm in \cite{zheng2015convergent,tu2016low} is very sensitive in the presence of even a single outlier, as the outliers can perturb the search directions arbitrarily. To handle outliers, existing convex optimization approaches based on sparse and low-rank decompositions can be applied using semidefinite programming \cite{li2017lowrank, wright2013compressive}. However, their computational cost is very expensive. Therefore, our goal in this paper is to develop fast and robust nonconvex alternatives that are globally convergent in a provable manner that can handle a large number of adversarial outliers.

\subsection{Our Approach and Results}

We propose a median-truncation strategy to robustify the gradient descent approach in \cite{zheng2015convergent,tu2016low}, which includes careful modifications on both initialization and local search. As it is widely known, the sample median is a more robust quantity to outliers, compared with the sample mean, which cannot be perturbed arbitrarily unless over half of the samples are outliers \cite{huber2011robust}. Therefore, it becomes an ideal metric to illuminate samples that are likely to be outliers and therefore should be eliminated during the gradient descent updates. Indeed, in a recent work by a subset of current authors \cite{zhang2016provable}, a median-truncated gradient descent algorithm has been proposed to robustify phase retrieval via a nonconvex method, where the sample median was exploited to control both the initialization and the local search step, so that only a subset of samples are selected to contribute to the search direction in each iteration. It was demonstrated that such an approach provably tolerates a constant fraction of outliers at a near-optimal sample complexity up to a logarithmic factor for the phase retrieval problem.

Inspired by \cite{zhang2016provable}, we design a tailored median-truncated gradient descent (median-TGD) algorithm for low-rank matrix recovery, where we carefully set the truncation strategy to mitigate the impact of outliers. Specifically, we develop a truncated spectral method for initialization, where only samples whose absolute values are not too deviated from the sample median are included. Similarly, we develop a truncated gradient update, where only samples whose measurement residuals using the current estimates are not too deviated from the sample median are included. This leads to an adaptive, iteration-varying strategy to mitigate the effects of outliers. In particular, the proposed algorithm does not assume a priori information regarding the outliers in terms of their fraction, distribution nor values. 

Theoretically, we demonstrate that, when initialized in a basin of attraction close to the ground truth, the proposed algorithm converges to the ground truth at a linear rate for the Gaussian measurement model with an order of $nr\log{n}$ measurements, where $n =(n_1+n_2)/2$, even when a constant fraction of the measurements are arbitrarily corrupted, which is nearly optimal up to a logarithmic factor. In addition, the truncated spectral method ensures an initialization in the basin of attraction with an order of $nr^2\log{n}\log^2 r$ measurements when a fraction of $1/\sqrt{r}$ measurements are arbitrarily corrupted. In the case when the rank is a small constant, our results indicate that the proposed algorithm can tolerate a constant fraction of outliers with an order of $n\log n$ measurements, which is much smaller than the size of the matrix.



To obtain the performance guarantees, we establish that the proposed median-truncated gradient satisfies a so-called regularity condition \cite{candes2015phase}, which is a sufficient condition for establishing the linear convergence to the ground truth. Since its debut in \cite{candes2015phase}, the regularity condition has been employed successfully in the analysis of phase retrieval \cite{candes2015phase,zhang2016reshaped,chen2015solving,zhang2016provable}, blind deconvolution \cite{li2016rapid} and low-rank matrix recovery \cite{zheng2015convergent,tu2016low,park2016provable} in the recent literature, to name a few. However, our analysis is significantly more involved due to the fact that the truncation procedure involving low-rank matrices has not been tackled in the previous literature. In particular, we establish a new restricted isometry property (RIP) of the sample median for the class of low-rank matrices, which can be thought as an extension of the RIP for the sample mean in compressed sensing literature \cite{recht2010guaranteed,candes2011tight}. We remark that such a result can be of independent interest, and its establishment is non-trivial due to the nonlinear character of the median operation. Numerical experiments demonstrate the excellent empirical performance of the proposed algorithm for low-rank matrix recovery from outlier-corrupted measurements, which significantly outperforms the existing algorithms that are not resilient to outliers \cite{zheng2015convergent,tu2016low}.

Computationally, because the sample median can be computed in a linear time \cite{tibshirani2008fast}, our median-truncated gradient descent algorithm shares a similar attractive computational cost as \cite{zheng2015convergent, tu2016low}. Specifically, the per-iteration computational complexity of the proposed algorithm is on the order of $\mathcal{O}\left(mn^{2} + 2n^{2}r + 4nr^{2}\right)$, which is linear with respect to $m$, while is quadratic with respect to $n$ and $r$\footnote{In practice, our algorithm can be applied to other measurement ensembles with more structures, such as sparsity, and the computational complexity can be further reduced.}. The proposed algorithm enjoys a lower computational complexity, compared with SVD-based methods \cite{jain2010guaranteed} and alternating minimization \cite{jain2013low}, which usually require more than $\mathcal{O}\left(mn^{2}  + n^{3}\right)$ or $\mathcal{O}\left(mn^{2}  + m^2\right)$ operations during each iteration. 

\subsection{Related Works}

Our work is amid the recent surge of nonconvex approaches for high-dimensional signal estimation, e.g.\ an incomplete and still growing list \cite{candes2015phase,zhao2015nonconvex,zhang2016reshaped,chen2015fast,chen2015solving,zhang2016provable, li2016rapid,zheng2015convergent,tu2016low,park2016provable}. The most closely-related work is on low-rank matrix recovery using random linear measurements  \cite{zheng2015convergent,tu2016low} in the absence of outliers, in the context of which our algorithm can be thought as a robust counterpart. Our particular approach is inspired by the previous work of a subset of current authors \cite{zhang2016provable} on robust phase retrieval, which can be thought as robust recovery of a rank-one positive semidefinite (PSD) matrix using rank-one measurement operators \cite{chen2015exact}. Our model in the current paper differs as we tackle low-rank matrix recovery using random full-rank  measurement operators, and thus non-trivial technical developments are necessary.

It is worth mentioning that other nonconvex approaches for robust low-rank matrix completion have been presented in \cite{yi2016fast,cherapanamjeri2016nearly,zhang2017nonconvex}, where the goal is to separate a low-rank matrix and sparse outliers from a small number of direct or linear measurements of their sum. The approaches typically use thresholding-based truncation for outlier removal and projected gradient descent for low-rank matrix recovery, which are somewhat similar to our approach in terms of different ways to remove outliers. However, this line of work typically requires stronger assumptions on the outliers such as spread-ness conditions, while we allow arbitrary outliers.


\subsection{Paper Organizations and Notations}
The remainder of this paper is organized as follows. Section~\ref{sec_problem} formulates the problems of interest. Section~\ref{sec_main} describes the proposed algorithm and its performance guarantees. Section~\ref{sec_numerical} provides numerical evidence on the superior performance of the proposed algorithm in the presence of outliers. Section~\ref{sec_proof} and Section~\ref{sec_initialization} provide the proofs of the main theoretical results, and finally, we conclude the paper in Section~\ref{sec_conclusions}.

Throughout this paper, we denote vectors by boldface lowercase letters and matrices by boldface uppercase letters. The notations $\bA^{T}$, $\left\Vert \bA \right\Vert$, and $\left\Vert \bA \right\Vert_{F}$ represent the transpose, the spectral norm and the Frobenius norm of a matrix $\bA$, respectively. We denote the $k$th singular value of $\bA$ by $\sigma_k(\bA)$, and the $k$th eigenvalue by $\lambda_{k}(\bA)$. For a vector $\by\in\mathbb{R}^n$, $\mathrm{med}(\by)$ denotes the median of the entries in $\by$, and $|\by|$ denotes the vector that contains its entry-wise absolute values. The $(k,t)$th entry of a matrix $\bA$ is denoted by $\bA_{k,t}$. Besides, the inner product between two matrices $\bA$ and $\bB$ is defined as $\langle \bA, \bB \rangle = \mathrm{Tr}\left(\bB^{T}\bA\right)$, where $\mathrm{Tr}(\cdot)$ denotes the trace of a matrix. The indicator function of an event $\mathcal{E}$ is denoted by $\mathbb{I}_{\mathcal{E}}$, which equals to $1$ if $\mathcal{E}$ is true and $0$ otherwise. In addition, we use $C$, $c_1$, $c_2$, $\ldots$ with different superscripts and subscripts to represent universal constants, whose values may change from line to line.
 

\section{Problem Formulation} \label{sec_problem}

Let $\bM\in\mathbb{R}^{n_{1}\times n_{2}}$ be a rank-$r$ matrix that can be written as 
\begin{equation}
\bM=\bX\bY^T, 
\end{equation}
where $\bX\in\mathbb{R}^{n_{1}\times r}$ and $\bY\in\mathbb{R}^{n_{2}\times r}$ are the low-rank factors of $\bM$. Define the condition number and the {\em average} condition number of $\bM$ as $\kappa = \frac{\sigma_1(\bM)}{\sigma_r(\bM)}$, and $\bar{\kappa} =\frac{\|\bM\|_F}{\sqrt{r}\sigma_r(\bM)}$,
respectively. Clearly, $\bar{\kappa}\leq \kappa$. 


Let $m$ be the number of measurements, and the set of sensing matrices are given as $\{\bA_i\}_{i=1}^m$, where $\bA_{i} \in \mathbb{R}^{n_{1}\times n_{2}}$ is the $i$th sensing matrix. In particular, each entry of $\bA_i$ is generated with i.i.d. standard Gaussian entries, i.e. $\left(\bA_{i}\right)_{k,t}\sim\mathcal{N}\left(0,1\right)$. Denote the index set of corrupted measurements by $\mathcal{S}$, and correspondingly, the index set of clean measurements is given as the complementary set $\mathcal{S}^c$. Mathematically, the measurements $\by=\{y_i\}_{i=1}^m$ are given as
\begin{equation}\label{measurement_model}
y_{i} = \left\{\begin{array}{cc}
\langle \bA_{i}, \bM \rangle , &\ \quad \mbox{if}\quad i \in\cS^c; \\ 
\eta_i  , & \quad \mbox{if}\quad i\in\cS, 
\end{array} \right.
\end{equation}
where $\boldsymbol{\eta}=\{\eta_i\}_{i\in\mathcal{S}}$ is the set of outliers that can take arbitrary values. Denote the cardinality of $\mathcal{S}$ by $|\mathcal{S}| =s\cdot m$, where $0\leq s<1$ is the fraction of outliers. To simplify the notations, we define the linear maps $\mathcal{A}_i(\bM)=\{\mathbb{R}^{n_{1}\times n_{2}} \mapsto \mathbb{R}: \langle \bA_i, \bM\rangle\} $, and $\mathcal{A}(\bM)=\{\mathbb{R}^{n_{1}\times n_{2}}\mapsto \mathbb{R}^m: \{\mathcal{A}_i(\bM)\}_{i=1}^m\}$. 




Instead of recovering $\bM$, we aim to directly recover its low-rank factors $\left(\bX, \bY\right)$ from the corrupted measurements $\by$, without a priori knowledge of statistical distribution or fractions of the outliers, in a computationally efficient and provably accurate manner. It is straightforward to see that for any orthonormal matrix $\bP\in\mathbb{R}^{r\times r}$ and scaler $\gamma \in\mathbb{R}$ such that $\gamma \neq 0$, we have $(\gamma\bX\bP) (\gamma^{-1}\bY\bP )^{T} = \bX\bY^{T}$. To address the scaling ambiguity, we assume $\bX^{T}\bX = \bY^{T}\bY$, and consequently, $\left(\bX, \bY\right)$ can be recovered only up to orthonormal transformations. Hence, we measure the estimation accuracy by taking this into consideration. Let the estimates of low-rank factors be $\bU\in\mathbb{R}^{n_{1}\times r}$ and $\bV\in\mathbb{R}^{n_{2}\times r}$, and define the augmented variables
\begin{equation}\label{augmented_variables}
\bW = \begin{bmatrix}\bU \\ \bV\end{bmatrix} \in\mathbb{R}^{\left(n_{1}+n_{2}\right)\times r}, \quad \bZ = \begin{bmatrix}\bX \\ \bY\end{bmatrix} \in\mathbb{R}^{\left(n_{1}+n_{2}\right)\times r}.
\end{equation}
Then the distance between $\bW$ and $\bZ$ is measured as
 \begin{equation}
\mathrm{dist}\left(\bW, \bZ\right) = \min_{\bP\in\mathbb{R}^{r\times r}, \bP\bP^{T} = \bI} \left\Vert \bW - \bZ\bP \right\Vert_{F}.
\end{equation}
Define
\begin{equation}\label{def_Q}
\bQ_{\left(\bW, \bZ\right)} = \argmin_{\bP\in\mathbb{R}^{r\times r}, \bP\bP^{T}  = \bI} \left\Vert \bW - \bZ\bP \right\Vert_{F},
\end{equation} 
and then $\mathrm{dist}\left(\bW, \bZ\right) = \left\Vert \bW - \bZ\bQ \right\Vert_{F}$, where the subscript of $\bQ$ is dropped for notational simplicity.


\section{Proposed Algorithm and Theoretical Guarantees} \label{sec_main}

Define a quadratic loss function with respect to the $i$th measurement as
\begin{equation}
f_i (\bU, \bV) =\frac{1}{4m}  \left(y_{i} -\mathcal{A}_i(\bU\bV^{T})  \right)^2,
\end{equation}
where $\bU\in\mathbb{R}^{n_{1}\times r}$ and $\bV\in\mathbb{R}^{n_{2}\times r}$. In order to get rid of the impact of outliers, an ideal approach is to minimize an {\em oracle} loss function, expressed as 
\begin{equation}\label{oracle_loss}
h_{\mathrm{oracle}}(\bU, \bV) = f_{\mathrm{oracle}}(\bU, \bV) + g(\bU,\bV) =  \sum_{i\in\mathcal{S}^c} f_i (\bU,\bV)  + \frac{\lambda}{4} \left\Vert \bU^{T}\bU - \bV^{T}\bV \right\Vert_{F}^{2},
\end{equation} 
which aims to minimize the quadratic loss over only the {\em clean} measurements, in addition to a regularization term
\begin{equation}
g(\bU,\bV) = \frac{\lambda}{4} \left\Vert \bU^{T}\bU - \bV^{T}\bV \right\Vert_{F}^{2},
\end{equation}
that aims at balancing the norm of the two factors. Nevertheless, it is impossible to minimize $h_{\mathrm{oracle}}(\bU,\bV)$ directly, since the oracle information regarding the support of outliers is absent. Moreover, the loss function is nonconvex, adding difficulty to its global optimization.

\subsection{Median-Truncated Gradient Descent}
We consider a gradient descent strategy where in each iteration, only a subset of all samples contribute to the search direction:
\begin{equation} \label{update_rule}
\begin{split}
\bU_{t+1} & = \bU_{t} - \frac{\mu_t}{\|\bU_0\|^{2}} \cdot \nabla_{\bU} h_{t}(\bU_{t}, \bV_{t}) ;\\
\bV_{t+1} & = \bV_{t} - \frac{\mu_t}{\|\bV_0\|^{2}} \cdot \nabla_{\bV} h_{t}(\bU_{t}, \bV_{t}) ,
\end{split}
\end{equation}
where $\mu_t$ denotes the step size, and $\bW_0=[\bU_{0}^T, \bV_{0}^T]^T$ is the initialization that will be specified later. Also, denote $\bW_t = [\bU_t^T,\bV_t^T]^T$. In particular, the iteration-varying loss function is given as 
\begin{equation}
h_{t}(\bU, \bV) =   \sum_{i\in  \mathcal{E}^t } f_i(\bU, \bV) + g(\bU,\bV) := f_{\mathrm{tr}} (\bU, \bV) + g(\bU,\bV) , 
\end{equation}
where the set $\cE^t$ varies at each iteration and includes only samples that are likely to be inliers. Denote the residual of the $i$th measurement at the $t$th iteration by
\begin{equation}
r_i^t = y_i   - \mathcal{A}_i (\bU_t\bV_t^T), \quad i=1,2,\ldots,m,
\end{equation}
and $\boldsymbol{r}^t =[r_1^t, r_2^t,\cdots, r_m^t]^T = \by - \mathcal{A}(\bU_t\bV_t^T)$. Then the set $\cE^t $ is defined as
\begin{equation}
\cE^t = \left\{ i \Big| |r_i^t| \leq \alpha_h \cdot \mathrm{med}\{ |\boldsymbol{r}^t| \} \right\},
\end{equation}
where $\alpha_h$ is some small constant. In other words, only samples whose current absolute residuals are not too deviated from the sample median of the absolute residuals are included in the gradient update. As the estimate $(\bU_t, \bV_{t})$ gets more accurate, we expect that the set $\cE^t$ gets closer to the oracle set $\cS^c$, and hence the gradient search is more accurate. Note that the set $\mathcal{E}^t$ varies per iteration, and therefore can adaptively prune the outliers. The gradients of $h_t(\bU,\bV)$ with respect to $\bU$ and $\bV$ are given as
\begin{equation}
\begin{split}
\nabla_{\bU} h_{t}(\bU, \bV) & = \frac{1}{2m}  \sum_{i\in  \mathcal{E}^t } \left[ \mathcal{A}_{i}\left(\bU\bV^{T}\right) - y_{i} \right] \bA_{i}\bV + \lambda \bU \left(\bU^{T}\bU - \bV^{T}\bV\right);\\
\nabla_{\bV} h_{t}(\bU,\bV) & = \frac{1}{2m}  \sum_{i\in  \mathcal{E}^t } \left[ \mathcal{A}_{i}\left(\bU\bV^{T}\right) - y_{i} \right]  \bA_{i}^{T}\bU + \lambda \bV\left(\bV^{T}\bV - \bU^{T}\bU\right).
\end{split}
\end{equation}


\begin{algorithm}[t]

\caption{Median-Truncated Gradient Descent (median-TGD)}\label{alg_minimize_quad_residual}
  
  \textbf{Parameters:} Thresholds $\alpha_{y}$ and $\alpha_{h}$, step size $\mu_{t}$, average condition number bound $\bar{\kappa}_{0}$, and rank $r$.
  
\textbf{Input:} Measurements $\by = \left\{y_{i}\right\}_{i=1}^{m}$, and sensing matrices $\left\{ \bA_{i} \right\}_{i=1}^{m}$.

\textbf{Initialization:} 

1) Set $\by_{1} = \left\{y_{i}\right\}_{i=1}^{m_{1}}$ and $\by_{2} = \left\{y_{i}\right\}_{i=m_1+1}^{m}$, where $m_1=\lceil m/2\rceil$ and $m_2=m-m_1$.

2) Take the rank-$r$ SVD of the matrix 
\begin{equation}\label{equ_initial_matrix}
\bK = \frac{1}{m_{1}} \sum_{i=1}^{m_1} y_{i}\bA_{i} \mathbb{I}_{\left\{ \left\vert y_{i} \right\vert \le \alpha_{y} \cdot \mathrm{med}\left( \left\vert \by_2 \right\vert \right) \right\}},
\end{equation}
which is denoted by $\bC_{L}\boldsymbol{\Sigma}\bC_{R}^{T}:=$ rank-$r$ SVD of $\bK$, where $\bC_{L} \in \mathbb{R}^{n_{1}\times r}$, $\bC_{R} \in \mathbb{R}^{n_{2}\times r}$ and $\boldsymbol{\Sigma} \in \mathbb{R}^{r\times r}$.

3) Initialize $\bU_{0} = \bC_{L} \boldsymbol{\Sigma}^{1/2} $, $\bV_{0} = \bC_{R} \boldsymbol{\Sigma}^{1/2} $.

\textbf{Gradient Loop:} For $t = 0:1:T-1$ do 
\begin{equation*}
\begin{split}
\bU_{t+1}  & = \bU_{t}  -    \frac{\mu_{t}}{\left\Vert \bU_{0} \right\Vert^{2}} \cdot \left[\frac{1}{2m}  \sum_{i=1}^{m} \left( \mathcal{A}_{i}\left(\bU_{t}\bV_{t}^{T}\right) - y_{i} \right) \bA_{i}\bV_{t}  \mathbb{I}_{ \mathcal{E}_{i}^t}  +  \lambda \bU_{t} \left(\bU_{t}^{T}\bU_{t} - \bV_{t}^{T}\bV_{t}\right)\right] ;\\
\bV_{t+1} & = \bV_{t}  -    \frac{\mu_{t}}{\left\Vert \bV_{0} \right\Vert^{2}} \cdot \left[\frac{1}{2m}  \sum_{i=1}^{m} \left( \mathcal{A}_{i}\left(\bU_{t}\bV_{t}^{T}\right) - y_{i} \right)  \bA_{i}^{T}\bU_{t}  \mathbb{I}_{ \mathcal{E}_{i}^t} + \lambda \bV_{t}\left(\bV_{t}^{T}\bV_{t} - \bU_{t}^{T}\bU_{t}\right)\right],
\end{split}
\end{equation*}
where 
\begin{equation*} 
\mathcal{E}_{i}^t = \left\{ \left\vert y_{i} -  \mathcal{A}_i (\bU_{t}\bV_{t}^{T})  \right\vert \le \alpha_{h} \cdot \mathrm{med}\left(  \left\vert \by - \mathcal{A}\left( \bU_{t}\bV_{t}^{T} \right)     \right\vert  \right)  \right\}.
\end{equation*}
\textbf{Output:} $\hat{\bX}= \bU_{T}$, and $\hat{\bY}= \bV_{T}$.
\end{algorithm}

For initialization, we adopt a truncated spectral method, which uses the top singular vectors of a sample-weighted surrogate matrix, where again only the samples whose values do not significantly digress from the sample median are included. To avoid statistical dependence in the theoretical analysis, we split the samples by using the sample median of $m_2$ samples to estimate $\|\bM\|_F$, and then using the rest of the samples to construct the truncated surrogate matrix to perform a spectral initialization. In practice, we find that this sample split is unnecessary, as demonstrated in the numerical simulations.

The details of the proposed algorithm, denoted as median-truncated gradient descent (median-TGD), are provided in Algorithm~\ref{alg_minimize_quad_residual}, where 
the stopping criterion is simply set as reaching a preset maximum number of iterations. In practice, it is also possible to set the stopping criteria by examining the progress between iterations. In sharp contrast to the standard gradient descent approach that exploits all samples in every iteration \cite{zheng2015convergent}, both the initialization and the search directions are controlled more carefully in order to adaptively eliminate outliers, while maintaining a similar low computational cost. 

\subsection{Theoretical Guarantees}

Theorem~\ref{main_theorem} summarizes the performance guarantee of median-TGD in Algorithm~\ref{alg_minimize_quad_residual} for low-rank matrix recovery using Gaussian measurements in the presence of sparse arbitrary outliers, when initialized within a proper neighborhood around the ground truth.

\begin{theorem}[Exact recovery with sparse arbitrary outliers]\label{main_theorem}
Assume the measurement model \eqref{measurement_model}, where each $\bA_i$ is generated with i.i.d. standard Gaussian entries. Suppose that the initialization $\bW_0$ satisfies 
\begin{equation*}
\mathrm{dist}\left(\bW_{0}, \bZ\right) \le \frac{1}{24} \sigma_{r}\left(\bZ\right).
\end{equation*}
Recall that $\kappa = \frac{\sigma_1(\bM)}{\sigma_r(\bM)}$. Set $\alpha_h=6$ and $\lambda =  \mathbb{E}\left[ \xi^{2}  \mathbb{I}_{\left\{  \left\vert \xi \right\vert \le 0.65\alpha_{h}  \right\}} \right]/4 $ with $\xi \sim \mathcal{N}\left(0, 1\right)$. There exist some constants $0<s_0<1$, $c_0>1, c_1>1$ such that with probability at least $1 - e^{ -c_{1}m }$, if $s\le s_{0}$, and $m \ge c_1nr\log{n}$, there exists a constant $\mu \le \frac{1}{740}$, such that with $\mu_{t}\le\mu$, the estimates of median-TGD satisfy
\begin{equation*}
\mathrm{dist}\left(\bW_{t}, \bZ\right)\le \left(1 - \frac{\mu}{10\kappa}\right)^{t/2} \mathrm{dist}\left(\bW_{0}, \bZ\right).
\end{equation*}
\end{theorem}

Theorem~\ref{main_theorem} suggests that if the initialization $\bW_0$ lies in the basin of attraction, median-TGD converges to the ground truth at a linear rate as long as the number $m$ of measurements is on the order of $nr\log n$, even when a constant fraction of measurements are corrupted arbitrarily. In comparisons, the gradient descent algorithm by Tu et.al. \cite{tu2016low} achieves the same convergence rate in a similar basin of attraction, with an order of $nr$ measurements using outlier-free measurements. Therefore, our algorithm achieves robustness up to a constant fraction of outliers with a slight price of an additional logarithmic factor in the sample complexity.

Theorem~\ref{theorem_initialization} guarantees that the proposed truncated spectral method provides an initialization in the basin of attraction with high probability.
\begin{theorem}\label{theorem_initialization}
Assume the measurement model \eqref{measurement_model}, and $\bar{\kappa} \leq \bar{\kappa}_{0}$. Set $\alpha_{y} = 2\log{(r^{1/4}\bar{\kappa}_{0}^{1/2} + 20)}$. There exist some constants $0<s_1<1$ and $c_2,c_3,c_4>1$ such that with probability at least $1 - n^{-c_2} -\exp(-c_3 m)$, if $s \le s_{1}/ (\sqrt{r}\bar{\kappa})$, and $m \ge c_4 \alpha_y^2 \bar{\kappa}^2 nr^{2}\log n $, we have
\begin{equation*}
\mathrm{dist}\left(\bW_{0}, \bZ\right) \le \frac{1}{24} \sigma_{r}\left(\bZ\right).
\end{equation*}
\end{theorem}
Theorem~\ref{theorem_initialization} suggests that the proposed initialization scheme is guaranteed to obtain a valid initialization in the basin of attraction with an order of $nr^2\log{n}\log^2r$ measurements when a fraction of $1/\sqrt{r}$ measurements are arbitrarily corrupted, assuming the average condition number $\bar{\kappa}$ is a small constant. In comparisons, in the outlier-free setting, Tu et.al. \cite{tu2016low} requires an order of $nr^2\kappa^2$ measurements for a one-step spectral initialization, which is closest to our scheme. Therefore, our initialization achieves robustness to a $1/\sqrt{r}$ fraction of outliers at a slight price of additional logarithmic factors in the sample complexity. It is worthwhile to note that in the absence of outliers, Tu et.al. \cite{tu2016low} was able to further reduce the sample complexity of initialization to an order of $nr$ by running multiple iterations of projected gradient descent. However, it is not clear whether such an iterative scheme can be generalized to the setting with outliers in our paper. 

Finally, we note that the parameter bounds in all theorems, including $\alpha_h$, $\alpha_y$ and $\mu$, are not optimized for performance, but mainly selected to establish the theoretical guarantees.

\section{Numerical Experiments} \label{sec_numerical}

In this section, we evaluate the performance of the proposed median-TGD algorithm via conducting several numerical experiments. As mentioned earlier, for the initialization step, in practice we find it is not necessary to split the samples into two parts. Therefore, the matrix in \eqref{equ_initial_matrix} is changed instead to
\begin{equation}\label{equ_initial_matrix_algorithm}
\bY = \frac{1}{m} \sum_{i=1}^{m} y_{i}\bA_{i} \mathbb{I}_{\left\{ \left\vert y_{i} \right\vert \le \alpha_{y} \cdot \mathrm{med}\left( \left\vert \by \right\vert \right) \right\}}.
\end{equation}
In particular, we check the trade-offs between the number of measurements, the rank and the fraction of outliers for accurate low-rank matrix recovery, and compare against the algorithm in \cite{tu2016low}, referred to as the vanilla gradient descent algorithm (vanilla-GD), to demonstrate the performance improvements in the presence of outliers due to median truncations. 

Let $n_{1} = 150$, $n_{2} = 120$. We randomly generate a rank-$r$ matrix as $\bM = \bX\bY^{T}$, where  both $\bX\in\mathbb{R}^{n_{1}\times r}$ and $\bY\in\mathbb{R}^{n_{2}\times r}$ are composed of i.i.d. standard Gaussian variables. The outliers are i.i.d. randomly generated following $\mathcal{N}(0, 10^{4}\left\Vert \bM \right\Vert_{F}^{2} )$. We set $\alpha_{y} = 12$ and $\alpha_{h} = 6$, and pick a constant step size $\mu_{t} = 0.4$. In all experiments, the maximum number of iterations for median-TGD algorithm is set as $T = 10^{4}$ to guarantee convergence. Moreover, let $(\hat{\bX}, \hat{\bY})$ be the solution to the algorithm under examination, and the recovered low-rank matrix is given as $\hat{\bM} =\hat{\bX}\hat{\bY}{}^{T}$. Then, the normalized estimate error is defined as $ \Vert \hat{\bM}  - \bM  \Vert_{F} / \left\Vert \bM \right\Vert_{F}$.

\subsection{Phase Transitions}

We first examine the phase transitions of median-TGD algorithm with respect to the number of measurements, the rank and the percent of outliers. Fix the percent of outliers as $s = 5 \%$. Given a pair of $r$ and $m$, a ground truth $(\bX, \bY)$ is generated composed of i.i.d. standard Gaussian variables. Multiple Monte Carlo trials are carried out, and each trial is deemed a success if the normalized estimate error is less than $10^{-6}$. Fig.~\ref{fig_Rectangle_est_random_matrix_outlier_fixnoutlier_changemr_outlier005} (a) shows the success rates of median-TGD, averaged over $10$ trials, with respect to the number of measurements and the rank, where the red line shows the theoretical limit defined as $r = (1-s)m/(n_{1}+n_{2})$ by a heuristic count of the degrees of freedom. It can be seen that the required number of measurements for a successful matrix recovery scales linearly with the rank $r$, and the transition is sharp. We next examine the success rates of median-TGD with respect to the percent of outliers and the rank. Fix $m = 2700$. Under the same setup as Fig.~\ref{fig_Rectangle_est_random_matrix_outlier_fixnoutlier_changemr_outlier005} (a), Fig.~\ref{fig_Rectangle_est_random_matrix_outlier_fixnoutlier_changemr_outlier005} (b) shows the success rate of median-TGD, averaged over $10$ trials, with respect to the rank and the percent of outliers. The performance of median-TGD degenerates smoothly with the increase of the percent of outliers. Similarly, the red line shows the theoretical limit as a comparison.
\begin{figure*}[htp]
\begin{tabular}{cc}
\includegraphics[width=0.5\textwidth]{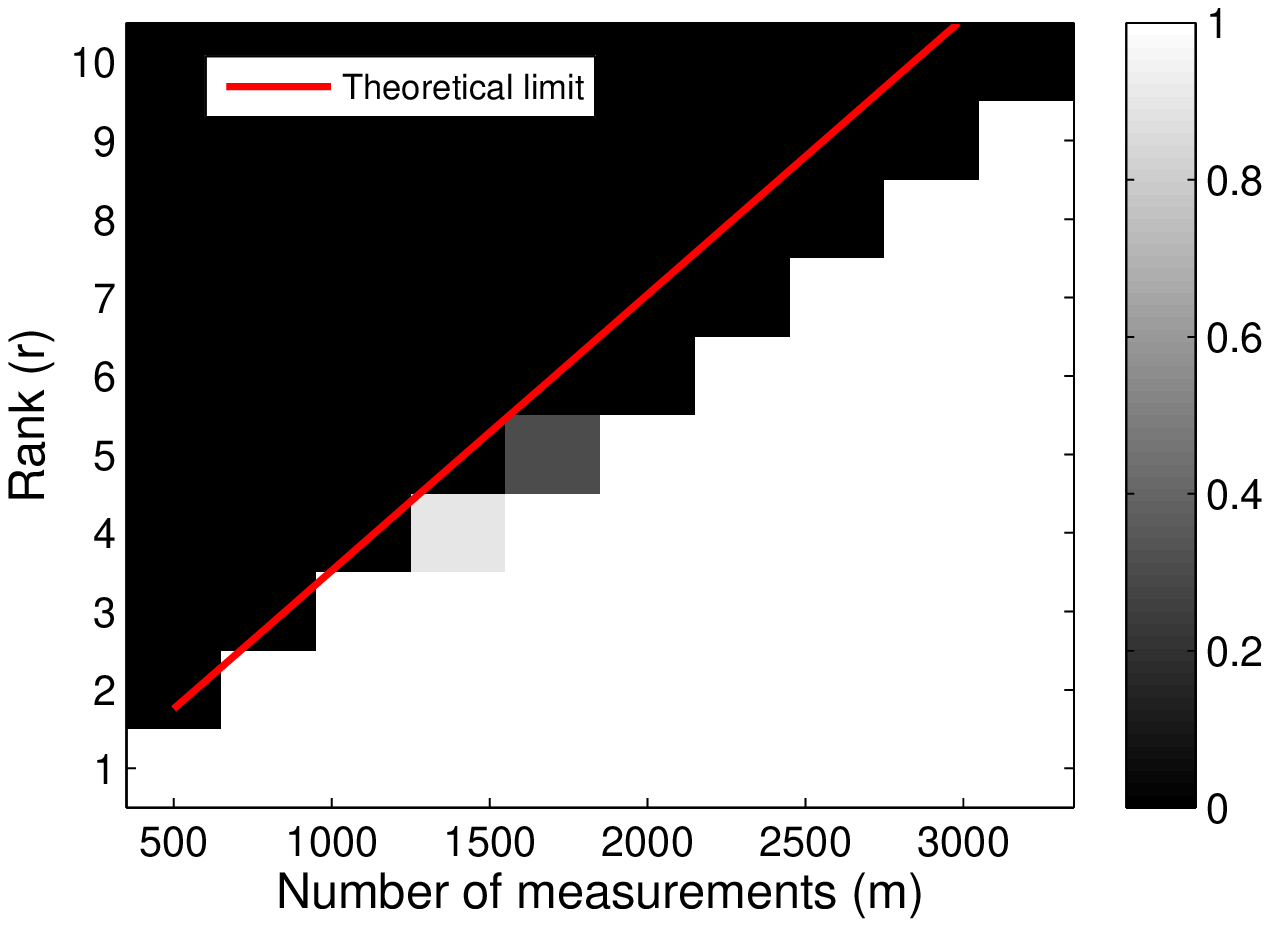} &\includegraphics[width=0.5\textwidth]{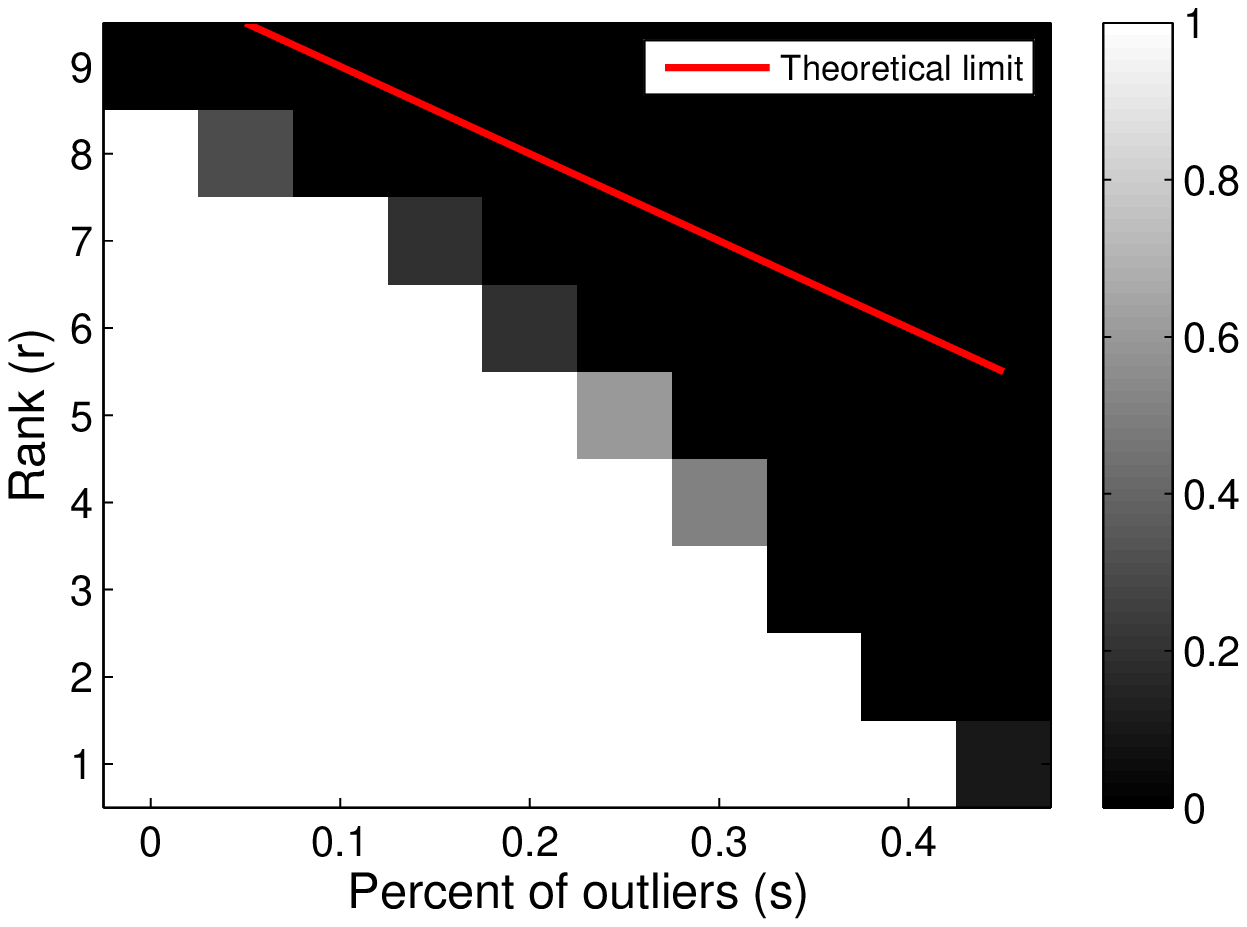}\\
(a) & (b)
\end{tabular}
\caption{Phase transitions of low-rank matrix recovery when $n_1=150$ and $n_2=120$. (a) Success rate with respect to the number of measurements and the rank, when 5$\%$ of measurements are corrupted by outliers. (b) Success rate with respect to the percent of outliers and the rank, when $m = 2700$.}
\label{fig_Rectangle_est_random_matrix_outlier_fixnoutlier_changemr_outlier005}
\end{figure*}

\subsection{Stability to Additional Bounded Noise}

We next examine the performance of median-TGD when the measurements are contaminated by both sparse outliers and dense noise. Here, the measurements are rewritten as
\begin{equation*}
y_{i} = \left\{\begin{array}{cc}
\langle \bA_{i}, \bM \rangle + w_{i} , &\quad \mbox{if}\quad i \in\cS^c; \\
\eta_i + w_{i},  & \quad \mbox{if}\quad i\in\cS,
\end{array} \right.
\end{equation*}
where $w_{i}$, for $i=1,2,\dots,m$, denote the additional bounded noise. 
\begin{figure}[ht]
\centering
\includegraphics[width=0.55\textwidth]{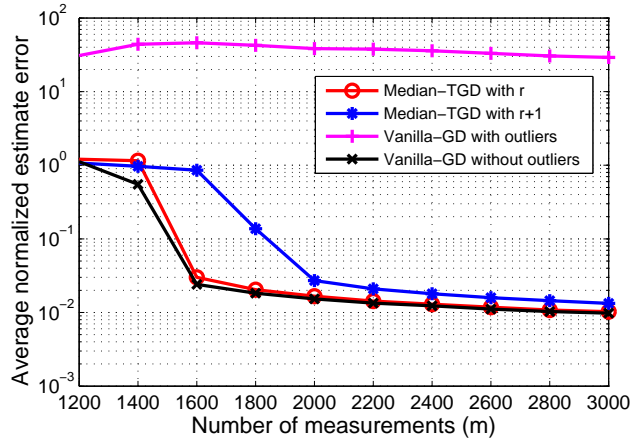}
\caption{Comparisons of average normalized estimate errors between median-TGD  and vanilla-GD in \cite{tu2016low} with respect to the number of measurements, with 5$\%$ of measurements corrupted by outliers and additional bounded noise, when $n_1=150$, $n_2=120$, and $r = 5$.}
\label{fig_Rectangle_est_random_matrix_outlier_noisy_comparison}
\end{figure}
Fix $r = 5$ and $s = 5 \%$. The dense noise is generated with i.i.d. random entries following $0.05\sigma_{5}\left(\bM\right) \cdot \mathcal{U}\left[-1, 1 \right]$. Fig.~\ref{fig_Rectangle_est_random_matrix_outlier_noisy_comparison} depicts the average normalized reconstruction errors with respect to the number of measurements using both median-TGD and vanilla-GD \cite{tu2016low}, where vanilla-GD is always given the true rank information, i.e. $r=5$. The performance of median-TGD is comparable to that of vanilla-GD using outlier-free measurements, which cannot produce reliable estimates when the measurements are corrupted by outliers. Therefore, median-TGD can handle outliers in a much more robust manner. Moreover, the performance of median-GD is stable as long as an upper bound of the true rank is used.

We next compare the convergence rates of median-TGD and vanilla-GD under various outlier settings, by fixing $m=2400$ while keeping the other settings the same as Fig.~\ref{fig_Rectangle_est_random_matrix_outlier_noisy_comparison}.
Fig.~\ref{fig_Rectangle_est_random_matrix_outlier_noisy_convergence_rate} shows the normalized estimate error with respect to the number of iterations of median-TGD  and vanilla-GD with no outliers, $1\%$ of outliers, and $10\%$ of outliers, respectively. In the outlier-free case, both algorithms have comparable convergence rates. However, even with a few outliers, vanilla-GD suffers from a dramatical performance degradation, while median-TGD is robust against outliers and can still converge to an accurate estimate. 
\begin{figure*}[htp]
\begin{tabular}{ccc}
\hspace{-0.15in}\includegraphics[width=0.35\textwidth]{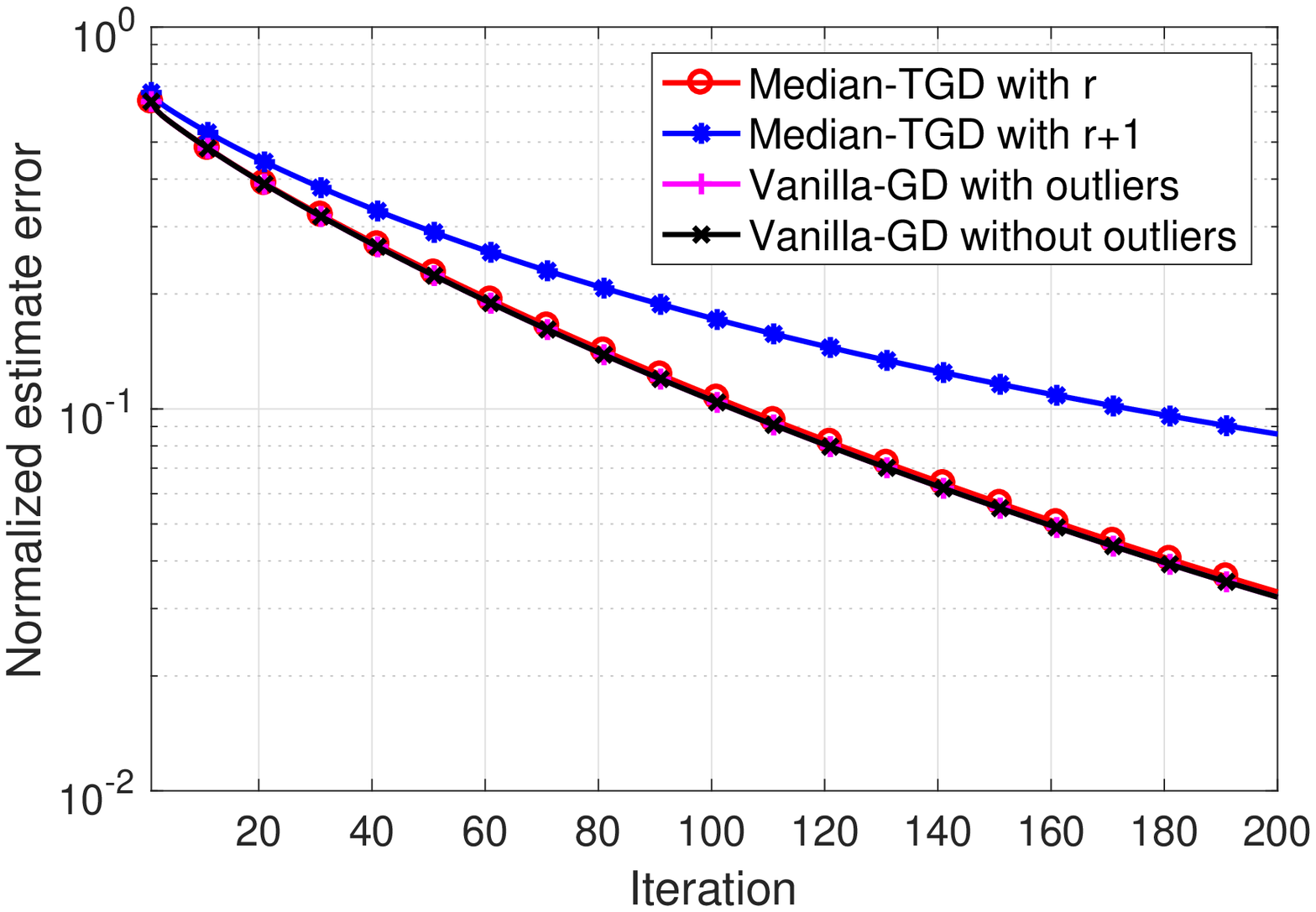} &\hspace{-0.2in}\includegraphics[width=0.35\textwidth]{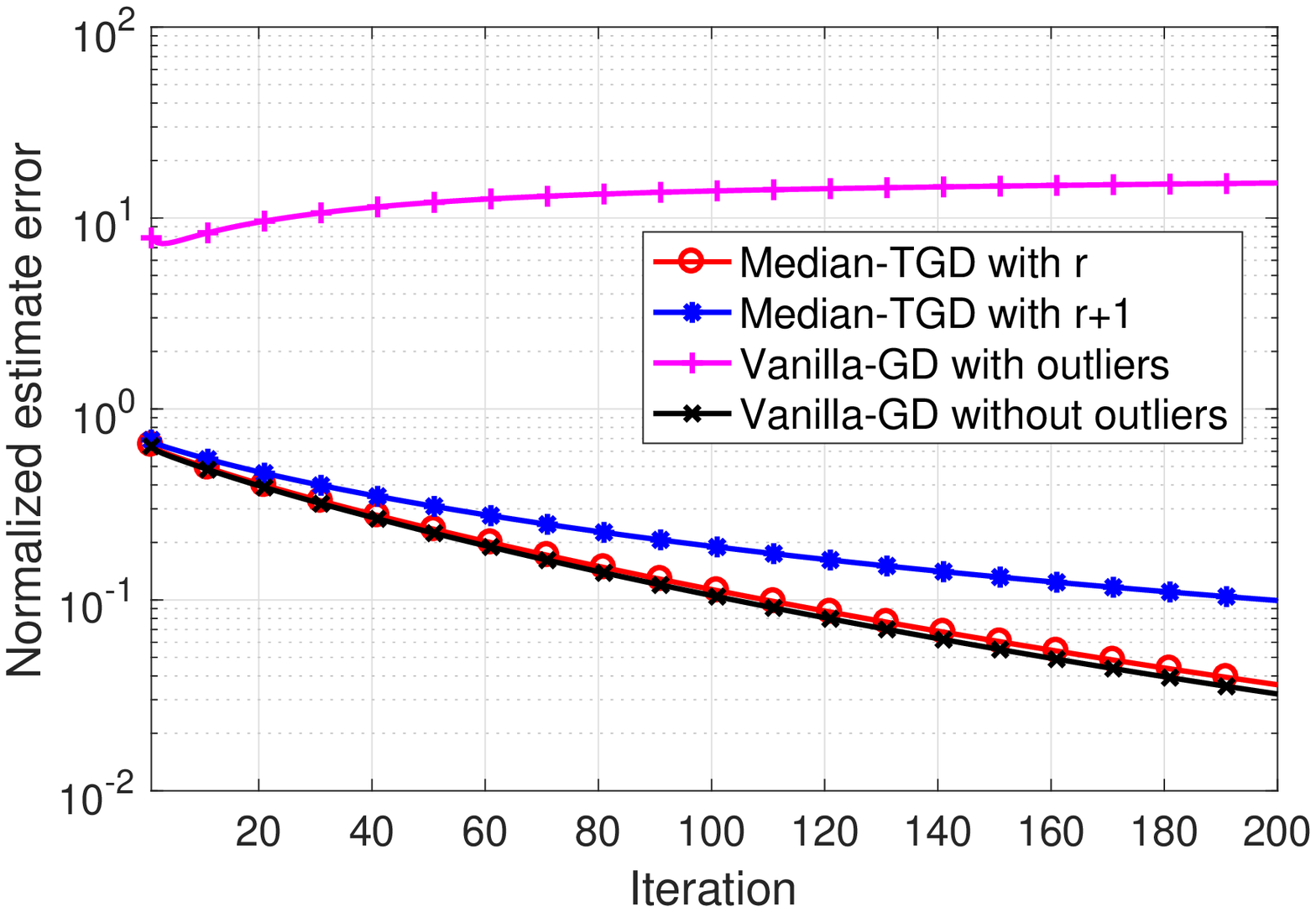}&\hspace{-0.2in}\includegraphics[width=0.35\textwidth]{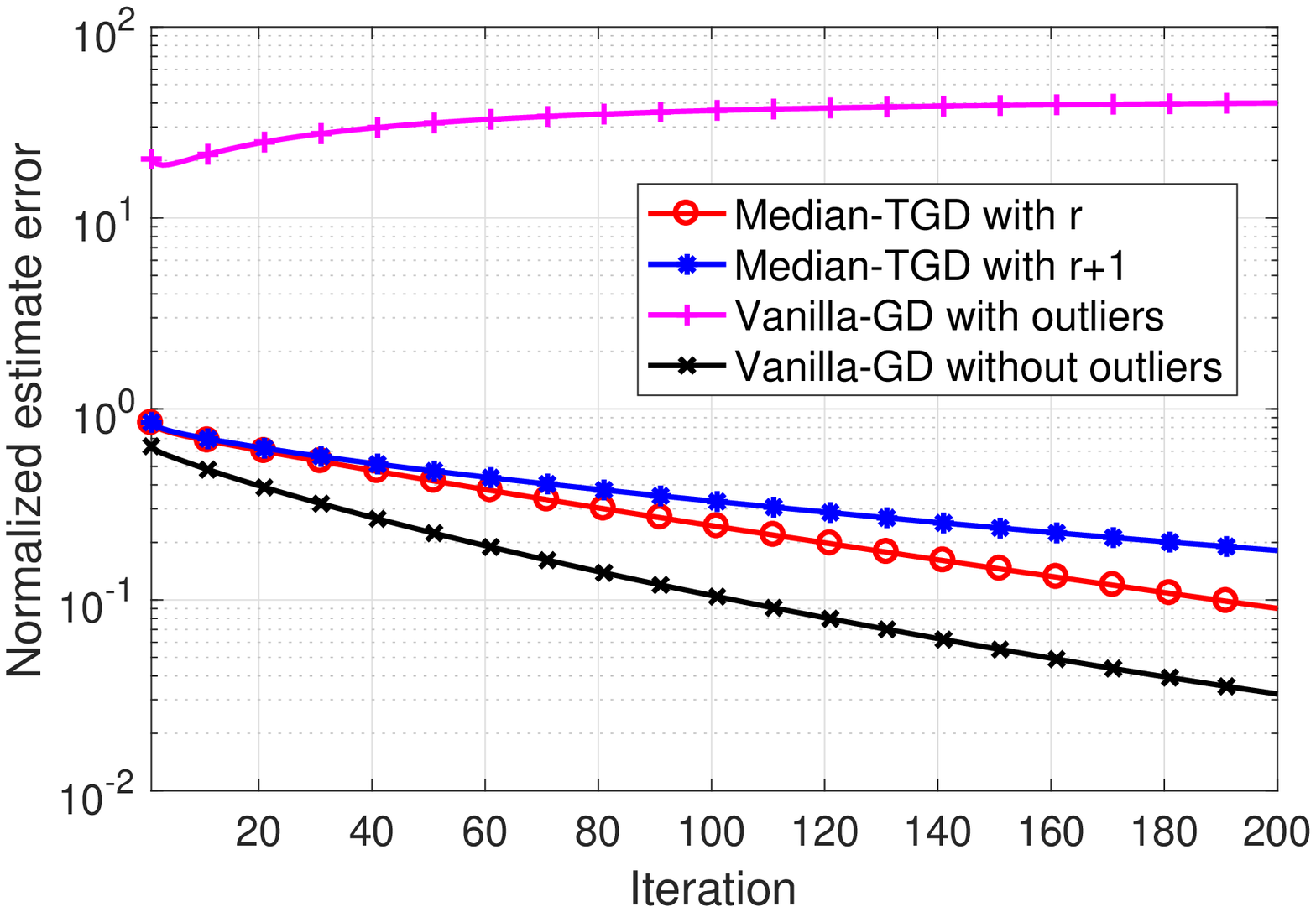}\\
\hspace{-0.15in} (a) $s=0$ &\hspace{-0.2in} (b) $s=1\%$ & \hspace{-0.2in} (c) $s=10\%$
\end{tabular}
\caption{The comparisons of convergence rates between median-TGD and vanilla-GD in different outlier-corruption scenarios, when $m = 2400$, $n_1=150$ and $n_2=120$.}
\label{fig_Rectangle_est_random_matrix_outlier_noisy_convergence_rate}
\end{figure*}

\section{Proof of Linear Convergence}\label{sec_proof}

In this section, we present the proof of Theorem~\ref{main_theorem}. Section~\ref{sec_RIP} first establishes an RIP-like property for the {\em median} of random linear measurements of low-rank matrices, which can be of independent interest. Section~\ref{sec_RC} describes the regularity condition (RC), which is used to certify the linear convergence of the proposed algorithm. Section~\ref{sec_certifyRC} proves several properties of the truncated gradient which are then used in Section~\ref{sec_certifyRC_outlier} to prove the RC and finish the proof.

\subsection{Concentration Property of Sample Median}\label{sec_RIP}

To begin, we define below the quantile function of a population distribution and its corresponding sample version.

\begin{definition}[Generalized quantile function] 
Let $0<\tau<1$. For a cumulative distribution function (CDF) $F(x)$, the generalized quantile function is defined as 
\begin{equation*}
F^{-1}\left(\tau\right) = \inf{\left\{x\in\mathbb{R}: F\left(x\right)\ge \tau \right\}}.
\end{equation*}
For simplicity, denote $\theta_{\tau}\left(F\right) = F^{-1}\left(\tau\right) $ as the $\tau$-quantile of $F$. Moreover, for a sample collection $\by = \left\{y_{i}\right\}_{i=1}^{m}$, the sample $\tau$-quantile $\theta_{\tau}\left(\by\right)$ means $\theta_{\tau} (\hat{F} )$, where $\hat{F}$ is the empirical distribution of the samples $\by$. Specifically, $\mathrm{med}\left(\by\right) = \theta_{1/2}\left(\by\right)$.
\end{definition}

We establish a RIP-style concentration property for the sample median used in the truncation indicator of gradient descent, which provides theoretical footings on the success of the proposed algorithm. The concentration property of the sample $p$-quantile function $\theta_{p}\left(  \left\vert  \mathcal{A}\left(\bG \right) \right\vert  \right)$ of all rank-$2r$ matrices $\bG$ is formulated in the following proposition, of which the proof is shown in Appendix~\ref{proof_prop_sample_median_concen}.

\begin{proposition}\label{prop_sample_median_concen}
Fix $\epsilon \in \left(0,1\right)$. If $m \ge c_{0} \left(\epsilon^{-2} \log{\epsilon^{-1}}\right) nr \log{n}$ for some large enough constant $c_{0}$, then with probability at least $1 - c_{1}\exp{\left( -c_{2}m\epsilon^{2} \right)}$, where $c_{1}$ and $c_{2}$ are some constants, we have for all rank-$2r$ matrices $\bG\in\mathbb{R}^{n_{1}\times n_{2}}$,
\begin{align*}
& \theta_{\frac{1}{2}}\left(  \left\vert  \mathcal{A}\left(\bG \right)    \right\vert  \right)  \in \left[0.6745 - \epsilon, 0.6745 + \epsilon \right] \left\Vert \bG \right\Vert_{F}, \\
& \theta_{0.49}\left(  \left\vert  \mathcal{A}\left(\bG  \right)  \right\vert  \right)  \in \left[ 0.6588 - \epsilon, 0.6588 + \epsilon \right] \left\Vert \bG \right\Vert_{F}, \\
& \theta_{0.51}\left(  \left\vert  \mathcal{A}\left(\bG \right)  \right\vert  \right)  \in \left[0.6903 - \epsilon,  0.6903 + \epsilon \right] \left\Vert \bG \right\Vert_{F}.
\end{align*}
\end{proposition}


Proposition~\ref{prop_sample_median_concen} suggests that as long as $m$ is on the order of $nr\log n$, the sample median $\theta_{\frac{1}{2}}\left(  \left\vert  \mathcal{A}\left(\bG \right)  \right\vert  \right) $ concentrates around a scaled $\left\Vert \bG \right\Vert_{F}$ for all rank-$2r$ matrices $\bG$, which resembles the matrix RIP in \cite{recht2010guaranteed}. Based on Proposition~\ref{prop_sample_median_concen}, provided that $m \ge c_{0}  nr \log{n}$ for some large enough constant $c_{0}$, setting $\bG= \bX\bY^{T}  -\bU\bV^{T}$, we have
\begin{equation}\label{median_concentration}
 \theta_{0.49}, \theta_{\frac{1}{2}}, \theta_{0.51}\left(  \left\vert  \mathcal{A}\left(\bX\bY^{T} \right)  - \mathcal{A}\left( \bU\bV^{T} \right)  \right\vert  \right) \in [0.65, 0.70] \left\Vert \bX\bY^{T} - \bU\bV^{T} \right\Vert_{F}
\end{equation}
holds with probability at least $1 - c_{1}\exp{\left( -c_{2}m \right)}$ for all $\bU\in\mathbb{R}^{n_{1}\times r}$, $\bV\in\mathbb{R}^{n_{2}\times r}$, $\bX\in\mathbb{R}^{n_{1}\times r}$, and $\bY\in\mathbb{R}^{n_{2}\times r}$. On the other end, due to Lemma~\ref{lemma_samp_quantile_bound_outlier}, we have 
\begin{equation*}
 \theta_{\frac{1}{2} - s}\left(  \left\vert   \mathcal{A}\left( \bX\bY^{T} \right) - \mathcal{A}\left( \bU\bV^{T} \right) \right\vert \right) \le \mathrm{med}\left( \left\vert \by - \mathcal{A}\left(\bU\bV^{T} \right) \right\vert \right) \le  \theta_{\frac{1}{2}+s}\left(  \left\vert   \mathcal{A}\left( \bX\bY^{T} \right) - \mathcal{A}\left( \bU\bV^{T} \right) \right\vert \right).
\end{equation*}
As a result, when the fraction of corruption satisfies $s\le 0.01$, the above equation together with \eqref{median_concentration} yields 
\begin{equation}\label{equ_sample_median_concen}
0.65 \left\Vert \bX\bY^{T} - \bU\bV^{T} \right\Vert_{F} \le \mathrm{med}\left( \left\vert \by - \mathcal{A}\left(\bU\bV^{T} \right) \right\vert \right) \le 0.70 \left\Vert \bX\bY^{T} - \bU\bV^{T} \right\Vert_{F}.
\end{equation}
Therefore, an important consequence is that the truncation event $\mathcal{E}_i$ satisfies  
\begin{equation}\label{truncation_bound}
\mathbb{I}_{ \left\{ \left\vert  \langle \bA_{i}, \bU\bV^{T} \rangle - y_i  \right\vert  \le 0.65 \alpha_{h}  \left\Vert \bU\bV^{T}- \bX\bY^{T} \right\Vert_{F}   \right\} }  \le   \mathbb{I}_{ \mathcal{E}_{i}} \le \mathbb{I}_{ \left\{ \left\vert  \langle \bA_{i}, \bU\bV^{T} \rangle - y_i  \right\vert \le 0.70 \alpha_{h}  \left\Vert \bU\bV^{T}- \bX\bY^{T} \right\Vert_{F}   \right\} }.
\end{equation} 


\subsection{Regularity Condition} \label{sec_RC}

We first introduce the so-called Regularity Condition (RC) \cite{candes2015phase, zheng2015convergent, tu2016low} that characterizes the benign curvature of the loss function around the ground truth, and guarantees the linear convergence of gradient descent to the ground truth. 

We first rewrite the loss function in terms of the augmented variables in \eqref{augmented_variables}. Denote the matrix $\bB_{i} = \begin{bmatrix} \boldsymbol{0} & \frac{1}{2}\bA_{i} \\  \frac{1}{2}\bA_{i}^{T}  & \boldsymbol{0} \end{bmatrix}$, and define $\mathcal{B}_{i}\left(\bW\bW^{T}\right) := \langle \bB_{i}, \bW\bW^{T} \rangle$ and $\mathcal{B}\left(\bW\bW^{T}\right) \allowbreak := \left\{ \mathcal{B}_{i}\left(\bW\bW^{T}\right) \right\}_{i=1}^{m}$, then we can have the equivalent representation 
\begin{equation}
\mathcal{B}_{i}\left(\bW\bW^{T}\right) = \langle \bB_{i}, \bW\bW^{T} \rangle = \langle \bA_{i}, \bU\bV^{T} \rangle = \mathcal{A}_{i}\left(\bU\bV^{T}\right).
\end{equation}
The regularizer can be rewritten as
\begin{equation}
g\left(\bW\right) = \frac{\lambda}{4} \left\Vert \bW^{T}\bD \bW\right\Vert_{F}^{2},
\end{equation}
where $\bD = \begin{bmatrix} \bI_{n_{1}} & \boldsymbol{0} \\ \boldsymbol{0} & -\bI_{n_{2}}\end{bmatrix}$, and its gradient can be rewritten as
\begin{equation}
\nabla g(\bW) = \lambda \bD\bW\left(\bW^{T}\bD\bW\right).
\end{equation}
Then the truncated gradient can be rewritten as a function of $\bW$,
\begin{equation}
\nabla h(\bW) = \frac{1}{m} \sum_{i=1}^{m} \left( \mathcal{B}_{i}\left(\bW\bW^{T}\right) - y_{i}\right) \bB_{i}\bW \mathbb{I}_{\mathcal{E}_{i}} + \lambda \bD\bW\left(\bW^{T}\bD\bW\right) := \nabla f_{tr} (\bW) +  \nabla g(\bW),
\end{equation}
where 
\begin{equation}
\mathcal{E}_{i} = \left\{ \left\vert y_{i} -  \mathcal{B}_i (\bW\bW^{T})  \right\vert \le \alpha_{h} \cdot \mathrm{med}\left(  \left\vert \by - \mathcal{B}\left( \bW\bW^{T} \right)     \right\vert  \right)  \right\}.
\end{equation}

Then the RC is defined in the following definition.

\begin{definition}[Regularity Condition]
Suppose $\bZ \in \mathbb{R}^{(n_{1}+n_{2})\times r}$ is the ground truth. The set of matrices that are in an $\epsilon$-neighborhood of $\bZ$ is defined as
\begin{equation*}
\mathcal{C}\left(\epsilon\right) = \left\{ \bW\in\mathbb{R}^{(n_{1}+n_{2})\times r}: \mathrm{dist}\left(\bW, \bZ\right)\le\epsilon \right\}.
\end{equation*}
Then the function $h(\bW)$ is said to satisfy the RC, denoted by $\mathrm{RC}\left(\alpha, \beta, \epsilon\right)$, if for all matrices $\bW\in \mathcal{C}\left(\epsilon\right)$, the following inequality holds:
\begin{equation}\label{def_RC}
\langle \nabla h\left(\bW\right), \bW - \bZ\bQ \rangle  \ge \frac{\sigma_{r}^{2}\left(\bZ\right)}{\alpha} \left\Vert \bW-\bZ\bQ \right\Vert_{F}^{2}  + \frac{1}{\beta\left\Vert\bZ\right\Vert^{2}} \left\Vert \nabla h\left(\bW\right)  \right\Vert_{F}^{2},
\end{equation} 
where $\bQ$ is an orthonormal matrix given in \eqref{def_Q}.
\end{definition}

The neighborhood $\mathcal{C}(\epsilon)$ is known as the basin of attraction. Interestingly, if $h(\bW)$ satisfies the RC, then initializing a simple gradient descent algorithm in the basin of attraction guarantees that the iterates converge at a linear rate to the ground truth, as summarized in the following lemma. 

\begin{lemma}\cite{candes2015phase, zheng2015convergent, tu2016low}\label{lemma_RC}
Suppose that $h(\bW)$ satisfies $\mathrm{RC}\left(\alpha, \beta, \epsilon\right)$ and $\bW_{0} \in \mathcal{C}\left(\epsilon\right)$. Consider the gradient descent update 
\begin{equation}\label{update_surrogate}
\bW_{t+1} = \bW_{t} - \frac{\mu}{\left\Vert\bZ\right\Vert^{2}} \nabla h\left(\bW_{t}\right)
\end{equation}
with the step size $0<\mu<\min\left\{\alpha/2, 2/\beta\right\}$. Then for all $t\ge 0$, we have $\bW_{t} \in \mathcal{C}\left(\epsilon\right)$ and 
\begin{equation*}
\mathrm{dist}\left(\bW_{t}, \bZ\right)\le \left(1 - \frac{2\mu}{\alpha\kappa}\right)^{t/2} \mathrm{dist}\left(\bW_{0}, \bZ\right).
\end{equation*}
\end{lemma}

Note that since the initialization satisfies $\mathrm{dist}\left(\bW_{0},\bZ\right) \le \frac{1}{24}\sigma_{r}\left(\bZ\right)$, by the triangle inequality we can guarantee that
\begin{equation*}
\frac{23}{24}\left\Vert\bZ\right\Vert \le \left\Vert\bW_{0}\right\Vert \le \frac{25}{24}\left\Vert\bZ\right\Vert,
\end{equation*} 
which implies 
\begin{align*}
&\frac{23}{24\sqrt{2}}\left\Vert\bZ\right\Vert \le \left\Vert\bU_{0}\right\Vert \le \frac{25}{24\sqrt{2}}\left\Vert\bZ\right\Vert;\\
&\frac{23}{24\sqrt{2}}\left\Vert\bZ\right\Vert \le \left\Vert\bV_{0}\right\Vert \le \frac{25}{24\sqrt{2}}\left\Vert\bZ\right\Vert,
\end{align*}
where we use the fact $\left\Vert \bU_{0} \right\Vert = \left\Vert \bV_{0} \right\Vert = \left\Vert \bW_{0} \right\Vert / \sqrt{2}$. Therefore, instead of proving the linear convergence of the actual update size $\frac{\mu}{\left\Vert\bU_{0}\right\Vert^{2}}$ and $\frac{\mu}{\left\Vert\bV_{0}\right\Vert^{2}}$, we prove it for the step size $\frac{\mu}{\left\Vert\bZ\right\Vert^{2}}$ in \eqref{update_surrogate}, since they only differ by a constant scaling of $\mu$. Hence, the rest of the proof is to verify that RC holds for the truncated gradient.

\subsection{Properties of Truncated Gradient}\label{sec_certifyRC}

We start by proving a few key properties of the truncated gradient $\nabla h(\bW) = \nabla f_{tr} (\bW) +  \nabla g(\bW)$. Consider the measurement model with sparse outliers in \eqref{measurement_model}. 
Define the truncation event
\begin{equation*}
\tilde{\mathcal{E}}_{i} = \left\{ \left\vert \mathcal{B}_i (\bZ\bZ^{T}) -  \mathcal{B}_i (\bW\bW^{T})  \right\vert \le \alpha_{h}  \mathrm{med}\left(  \left\vert \by - \mathcal{B}\left( \bW\bW^{T} \right)     \right\vert  \right)  \right\},
\end{equation*}
which is the same as $\mathcal{E}_i$ except that the measurements used to calculate the residual are replaced by clean measurements. In particular, it is straight to see that \eqref{truncation_bound} also holds for $\tilde{\mathcal{E}}_{i} $. Then we can write $\nabla f_{tr}\left(\bW\right)$ as
\begin{align*}
\nabla f_{tr}\left(\bW\right)
& = \frac{1}{m} \sum_{i=1}^{m} \left( \mathcal{B}_{i}\left(\bW\bW^{T} \right) - y_{i}\right)\bB_{i}\bW\mathbb{I}_{ \mathcal{E}_{i}} \\
& = \frac{1}{m} \sum_{i \notin \mathcal{S} } \left( \mathcal{B}_{i}\left(\bW\bW^{T} \right) - \mathcal{B}_{i}\left(\bZ\bZ^{T} \right) \right)\bB_{i}\bW\mathbb{I}_{ \tilde{\mathcal{E}}_{i}}  +   \frac{1}{m} \sum_{i \in \mathcal{S}} \left( \mathcal{B}_{i}\left(\bW\bW^{T} \right) - y_{i}\right)\bB_{i}\bW\mathbb{I}_{ \mathcal{E}_{i}}\\
& = \underbrace{ \frac{1}{m} \sum_{i=1}^{m} \left( \mathcal{B}_{i}\left(\bW\bW^{T} \right) -  \mathcal{B}_{i}\left(\bZ\bZ^{T} \right) \right)\bB_{i}\bW\mathbb{I}_{ \tilde{\mathcal{E}}_{i}} }_{\nabla^{c} f_{tr}\left(\bW\right)}\\
& \quad +  \underbrace{ \frac{1}{m} \sum_{i \in \mathcal{S}} \left[ \left(\mathcal{B}_{i}\left(\bW\bW^{T} \right) - y_{i}\right)\mathbb{I}_{ \mathcal{E}_{i}}  -  \left(\mathcal{B}_{i}\left(\bW\bW^{T} \right) - \mathcal{B}_{i}\left(\bZ\bZ^{T} \right) \right)\mathbb{I}_{ \tilde{\mathcal{E}}_{i}}  \right] \bB_{i}\bW }_{\nabla^{o} f_{tr}\left(\bW\right)},
\end{align*}
where $\nabla^{c} f_{tr}\left(\bW\right)$ corresponds to the truncated gradient {\em as if} all measurements are clean, and $\nabla^{o} f_{tr}\left(\bW\right)$ corresponds to the contribution of the outliers.

For notational simplicity, define 
\begin{equation}\label{def_H}
\bH = \begin{bmatrix} \bH_{1} \\ \bH_{2} \end{bmatrix} =  \bW - \bZ\bQ = \begin{bmatrix} \bU - \bX\bQ  \\  \bV - \bY\bQ \end{bmatrix}, 
\end{equation}
where $\bQ$ is given in \eqref{def_Q}. We have
\begin{align}
\left\langle \nabla^c f_{tr}\left(\bW\right), \bH \right \rangle 
& =  \frac{1}{m} \sum_{i=1}^{m} \langle \bB_{i}, \bW\bW^{T} - \bZ\bZ^{T} \rangle  \cdot   \langle \bB_{i}, \bH\bW^{T} \rangle  \cdot  \mathbb{I}_{ \tilde{\mathcal{E}}_{i}} .\label{setD_split}
\end{align}
Define the set $\mathcal{D} $ as
\begin{align}\label{def_setD}
\mathcal{D} 
& = \left\{ i |  \langle \bB_{i}, \bW\bW^{T} - \bZ\bZ^{T}   \rangle  \cdot \langle \bB_{i}, \bH\bW^{T} \rangle  < 0  \right\} .
\end{align}
We can then split \eqref{setD_split} and bound it as
\begin{align}
  & \left\langle \nabla^c f_{tr}\left(\bW\right), \bH \right \rangle\\
  & \ge \frac{1}{m}\sum_{i \notin \mathcal{D}} \langle \bB_{i}, \bW\bW^{T} - \bZ\bZ^{T}   \rangle  \cdot \langle \bB_{i}, \bH\bW^{T} \rangle  \cdot    \mathbb{I}_{ \left\{   \left\vert  \langle \bB_{i}, \bW\bW^{T} - \bZ\bZ^{T}  \rangle   \right\vert   \le   0.65 \alpha_{h}  \left\Vert \bU\bV^{T}- \bX\bY^{T} \right\Vert_{F}  \right\}  } \nonumber \\
& \quad  + \frac{1}{m}\sum_{i \in \mathcal{D}} \langle \bB_{i}, \bW\bW^{T} - \bZ\bZ^{T}   \rangle  \cdot \langle \bB_{i}, \bH\bW^{T} \rangle  \cdot    \mathbb{I}_{ \left\{   \left\vert  \langle \bB_{i}, \bW\bW^{T} - \bZ\bZ^{T}  \rangle   \right\vert   \le   0.70 \alpha_{h}  \left\Vert \bU\bV^{T}- \bX\bY^{T} \right\Vert_{F}  \right\}  }  \nonumber\\
& = \underbrace{\frac{1}{2m}\sum_{i \notin \mathcal{D}} \langle \bA_{i}, \bU\bV^{T} - \bX\bY^{T}   \rangle  \cdot \langle \bA_{i}, \bH_{1}\bV^{T} + \bU\bH_{2}^{T} \rangle  \cdot    \mathbb{I}_{ \left\{   \left\vert  \langle \bA_{i}, \bU\bV^{T} - \bX\bY^{T}  \rangle   \right\vert   \le   0.65 \alpha_{h}  \left\Vert \bU\bV^{T}- \bX\bY^{T} \right\Vert_{F}  \right\}  }}_{B_1} \nonumber \\
& \quad  + \underbrace{\frac{1}{2m}\sum_{i \in \mathcal{D}} \langle \bA_{i}, \bU\bV^{T} - \bX\bY^{T}   \rangle  \cdot \langle \bA_{i}, \bH_{1}\bV^{T} + \bU\bH_{2}^{T} \rangle  \cdot    \mathbb{I}_{ \left\{   \left\vert  \langle \bA_{i}, \bU\bV^{T} - \bX\bY^{T}  \rangle   \right\vert   \le   0.70 \alpha_{h}  \left\Vert \bU\bV^{T}- \bX\bY^{T} \right\Vert_{F}  \right\}  }}_{B_2}  .\label{equ_rc_left_relax}
\end{align}

The first term in \eqref{equ_rc_left_relax} can be lower bounded by  Proposition~\ref{prop_local_curv_clean_first}, whose proof is given in Appendix~\ref{proof_prop_local_curv_clean_first}.
\begin{proposition}\label{prop_local_curv_clean_first}
Provided $m \ge c_{1}  nr $, then
\begin{align} \label{bound_B1}
B_1 
& \ge \frac{\gamma_{1}}{2}    \langle \bU\bV^{T} - \bX\bY^{T},  \bH_{1}\bV^{T} + \bU\bH_{2}^{T}  \rangle  -   0.0006 \alpha_{h} \left\Vert \bU\bV^{T}- \bX\bY^{T} \right\Vert_{F} \left\Vert \bH_{1}\bV^{T} + \bU\bH_{2}^{T}  \right\Vert_{F} 
\end{align}
holds for all $\bZ,\bW\in\mathbb{R}^{(n_{1}+n_{2})\times r}$ with probability at least $1 - \exp{\left( - c_2  m \right)}$, where $\gamma_{1} = \mathbb{E}\left[ \xi^{2}  \mathbb{I}_{\left\{  \left\vert \xi \right\vert \le 0.65\alpha_{h}  \right\}} \right] $ with $\xi \sim \mathcal{N}\left(0, 1\right)$, and $c_1,c_2>0$ are numerical constants.
\end{proposition}

The second term in \eqref{equ_rc_left_relax} can be lower bounded by Proposition~\ref{prop_local_curv_clean_second}, whose proof is given in Appendix~\ref{proof_prop_local_curv_clean_second}.

\begin{proposition}\label{prop_local_curv_clean_second}
Provided $m \ge c_{1}nr$, we have
\begin{align}\label{bound_B2}
B_2 &\ge -  0.36 \alpha_{h}  \left\Vert \bU\bV^{T}- \bX\bY^{T} \right\Vert_{F}  \left\Vert \bH_{1}\bH_{2}^{T} \right\Vert_{F}
\end{align}
holds for all $\bZ,\bW\in\mathbb{R}^{(n_{1}+n_{2})\times r}$ with probability at least $1-\exp\left(-c_{2}m\right)$, where $c_1, c_2>0$ are numerical constants.
\end{proposition}


The contribution of outliers $\nabla^{o} f_{tr}\left(\bW\right)$ can be bounded by the following proposition, whose proof is given in Appendix~\ref{proof_outlier_curvature}.
\begin{proposition} \label{prop_outlier_curvature}
Provided $m \ge c_{1}nr\log{n}$, we have
\begin{equation}\label{bound_outlier_curvature}
|\langle \nabla^{o} f_{tr}\left(\bW\right) ,\bH \rangle| \leq 0.71\alpha_{h}\sqrt{s} \left\Vert \bX\bY^{T} - \bU\bV^{T} \right\Vert_{F} \| \bH_{1}\bV^T + \bU\bH_{2}^{T}\|_F
\end{equation}
holds for all $\bZ,\bW\in\mathbb{R}^{(n_{1}+n_{2})\times r}$ with probability at least $1-\exp\left(-c_{2}m\right)$, where $c_1, c_2>0$ are numerical constants.
\end{proposition}

On the other end, Proposition~\ref{prop_local_smooth_clean} establishes an upper bound for $\left\Vert \nabla f_{tr}\left(\bW\right)  \right\Vert_{F}^{2}$, whose proof is given in Appendix~\ref{proof_prop_local_smooth_clean}.
\begin{proposition}\label{prop_local_smooth_clean}
Let $n = \max\{n_{1}, n_{2}\}$. Provided $m \ge c_{1}nr \log n$, we have
\begin{equation}\label{bound_spectral_gradient}
\left\Vert \nabla f_{tr}\left(\bW\right)  \right\Vert_{F}^{2} \le  0.25 \alpha_{h}^{2}  \left\Vert \bU\bV^{T} - \bX\bY^{T} \right\Vert_{F}^{2} \left\Vert \bW \right\Vert^{2}
\end{equation}
holds for all $\bZ,\bW\in\mathbb{R}^{(n_{1}+n_{2})\times r}$ with probability at least $1-\exp\left(-c_{2}m\right)$, where $c_1,c_2>0$ are numerical constants.
\end{proposition}

Moreover, for the regularizer, we have
\begin{align}
\langle \nabla g(\bW), \bH \rangle
& = \lambda  \langle \bD\bW\left(\bW^{T}\bD\bW\right), \bH\rangle \nonumber\\
& = \lambda \langle \bU\left(\bU^{T}\bU - \bV^{T}\bV\right), \bH_{1} \rangle  + \lambda \langle \bV\left(\bV^{T}\bV - \bU^{T}\bU\right), \bH_{2} \rangle \nonumber\\
& = \lambda \langle \bU\bU^{T}, \bH_{1}\bU^{T} \rangle + \lambda \langle \bV\bV^{T}, \bH_{2}\bV^{T}  \rangle - \lambda \langle \bU\bV^{T}, \bH_{1}\bV^{T} + \bU\bH_{2}^{T} \rangle, \label{eq_reg_g_curv}
\end{align}
and
\begin{align}
\left\Vert \nabla g(\bW) \right\Vert_{F}^{2} 
& = \lambda^{2} \left\Vert  \bD\bW\left(\bW^{T}\bD\bW\right) \right\Vert_{F}^{2} \nonumber\\
& = \lambda^{2} \left\Vert  \bW\left(\bW^{T}\bD\bW\right) \right\Vert_{F}^{2} \nonumber\\
& \le \lambda^{2} \left\Vert  \bW \right\Vert^{2} \left\Vert  \bW^{T}\bD\bW \right\Vert_{F}^{2}\nonumber\\
& = \lambda^{2} \left\Vert  \bW \right\Vert^{2} \left\Vert  \left(\bH + \bZ\bQ\right)^{T}\bD\left(\bH + \bZ\bQ\right) \right\Vert_{F}^{2}\nonumber\\
& = \lambda^{2} \left\Vert  \bW \right\Vert^{2} \left\Vert  \bH^{T}\bD\bH +  \bH^{T}\bD\bZ\bQ  +  \left( \bZ\bQ\right)^{T}\bD\bH \right\Vert_{F}^{2} \label{nablag_xx}\\
& \leq \lambda^{2} \left\Vert  \bW \right\Vert^{2} \left(\left\Vert  \bH^{T}\bD\bH \right\Vert_{F}  + 2\left\Vert  \bH^{T}\bD\bZ \right\Vert_{F} \right)^{2},  \label{equ_regu_g_local_smooth}
\end{align}
where \eqref{nablag_xx} follows from $\bX^{T}\bX = \bY^{T}\bY$.


\subsection{Certifying the RC with Sparse Outliers}\label{sec_certifyRC_outlier}

We are now ready to establish the RC in the neighborhood where $\left\Vert \bH \right\Vert_{F} \le \frac{1}{24} \sigma_{r}\left(\bZ\right) $. Recall that based on Propositions~\ref{prop_local_curv_clean_first}, \ref{prop_local_curv_clean_second} and \ref{prop_outlier_curvature}, and \eqref{eq_reg_g_curv}, we have 
\begin{align}
&\langle \nabla  h\left(\bW\right), \bW - \bZ\bQ \rangle \nonumber\\
&\ge  \langle \nabla f_{\mathrm{tr}} (\bW), \bH\rangle  + \langle  \nabla g(\bW), \bH \rangle \nonumber\\
& \ge  \langle \nabla f_{\mathrm{tr}}^{c} (\bW), \bH\rangle  - \left\vert \langle \nabla f_{\mathrm{tr}}^{o} (\bW), \bH\rangle \right\vert  + \langle  \nabla g(\bW), \bH \rangle   \nonumber\\
&\geq \frac{\gamma_1}{2}  \langle \bU\bV^{T} - \bX\bY^{T},  \bH_{1}\bV^{T} + \bU\bH_{2}^{T}   \rangle  -    0.0006 \alpha_{h} \left\Vert \bU\bV^{T}- \bX\bY^{T} \right\Vert_{F} \left\Vert \bH_{1}\bV^{T} + \bU\bH_{2}^{T}  \right\Vert_{F} \nonumber \\
& \quad -  0.36 \alpha_{h}  \left\Vert \bU\bV^{T}- \bX\bY^{T} \right\Vert_{F}  \left\Vert \bH_{1}\bH_{2}^{T} \right\Vert_{F} - 0.71\alpha_{h}\sqrt{s} \left\Vert \bU\bV^{T} - \bX\bY^{T} \right\Vert_{F} \| \bH_{1}\bV^T+\bU\bH_{2}^{T}\|_F\nonumber\\
& \quad + \langle  \nabla g(\bW), \bH \rangle. \label{RC_lhs}
\end{align}
Set $\alpha_{h} = 6$, we have $\gamma_{1} \approx 0.998348$. Set $\lambda = \gamma_{1}/4$, then we can write 
\begin{align}
&\quad \frac{\gamma_1}{2} \langle \bU\bV^{T} - \bX\bY^{T},  \bH_1\bV^T +\bU\bH_2^{T}   \rangle + \langle \nabla g(\bW), \bH\rangle \nonumber\\
& =  2\lambda  \langle \bU\bV^{T} - \bX\bY^{T},  \bH_1\bV^T +\bU\bH_2^{T}   \rangle +  \lambda \langle \bU\bU^{T}, \bH_{1}\bU^{T} \rangle + \lambda \langle \bV\bV^{T}, \bH_{2}\bV^{T}  \rangle\nonumber\\
&\quad - \lambda \langle \bU\bV^{T}, \bH_{1}\bV^{T} + \bU\bH_{2}^{T} \rangle\nonumber\\
&= \lambda \langle \bW\bW^T- \bZ\bZ^T,\bH\bW^T\rangle -\lambda\langle \bX\bY^T, \bH_1\bV^T+\bU\bH_2^T\rangle \nonumber  \\
&\quad + \lambda \langle \bX\bX^T, \bH_1\bU^T\rangle+\lambda \langle \bY\bY^T, \bH_2\bV^T\rangle,\label{positive_side}
\end{align}
where the last three terms can be re-arranged as 
\begin{align}
& \langle \bX\bX^T, \bH_1\bU^T\rangle + \langle \bY\bY^T, \bH_2\bV^T\rangle -  \langle \bX\bY^T, \bH_1\bV^T+\bU\bH_2^T\rangle  \nonumber \\
 & =   \langle (\bX\bQ)^T\bU, (\bX\bQ)^T\bH_1 \rangle +  \langle  (\bY\bQ)^T\bV, (\bY\bQ)^T \bH_2 \rangle \nonumber\\
 &\quad - \langle (\bY\bQ)^T\bV, (\bX\bQ)^T\bH_1 \rangle - \langle  (\bY\bQ)^T\bH_2, (\bX\bQ)^T \bU\rangle \nonumber\\
 & = \langle (\bY\bQ)^T\bV,  (\bY\bQ)^T \bH_2 - (\bX\bQ)^T\bH_1 \rangle +\langle  (\bX\bQ)^T\bH_1 - (\bY\bQ)^T\bH_2, (\bX\bQ)^T \bU\rangle   \nonumber \\
 & =  \langle (\bY\bQ)^T\bV -(\bX\bQ)^T\bU, (\bY\bQ)^T (\bV-\bY\bQ) - (\bX\bQ)^T(\bU-\bX\bQ) \rangle \nonumber \\
 & = \|  (\bY\bQ)^T\bV - (\bX\bQ)^T\bU\|_F^2 = \|\bH^{T}\bD\bZ\|_F^2, \label{relation_XY}
\end{align}
where \eqref{relation_XY} follows from $\bX^T\bX=\bY^T\bY$. Moreover, using the facts that $\bH^{T}\bZ\bQ$ and $\bH^T\bW$ are symmetric matrices and $\bW^{T}\bZ\bQ  \succeq 0$ \cite{sanghavi2016local}, we have the first term in \eqref{positive_side} bounded as
\begin{align}
 \langle \bW\bW^{T} - \bZ\bZ^{T}, & \bH\bW^{T}   \rangle  
  = \| \bH\bQ^T\bZ^T\|_F^2 + \| \bH^T\bZ\|_F^2 +  \|\bH\bH^T\|_F^2 + 3\langle \bH\bH^T, \bH\bQ^T\bZ^T \rangle \nonumber\\
 & \geq \| \bH\bQ^T\bZ^T\|_F^2 + \| \bH^T\bZ\|_F^2 +  \|\bH\bH^T\|_F^2 - 3 \|\bH\bH^T\|_F \| \bH\bQ^T\bZ^T\|_F \label{RC_lhs_line1}  \\
 & \geq  \| \bH\bQ^T\bZ^T\|_F^2 + \| \bH^T\bZ\|_F^2 +  \|\bH\bH^T\|_F^2 - \frac{1}{8} \| \bH\bQ^T\bZ^T\|_F^2 \label{RC_lhs_line2}\\
 &\geq \frac{7}{8} \| \bH\bQ^T\bZ^T\|_F^2 + \| \bH^T\bZ\|_F^2 +  \|\bH\bH^T\|_F^2\label{RC_lhs_term1}  ,   
\end{align}
where \eqref{RC_lhs_line1} follows from Cauchy-Schwarz inequality, \eqref{RC_lhs_line2} follows from $\left\Vert\bH\bH^{T}\right\Vert_{F} \le \left\Vert\bH\right\Vert_{F}^{2} \le \frac{1}{24}\sigma_{r}\left(\bZ\bQ\right)\left\Vert\bH\right\Vert_{F}\le \frac{1}{24}\left\Vert \bH\bQ^{T}\bZ^{T} \right\Vert_{F}$, where we used $\left\Vert \bH \right\Vert_{F} \le \frac{1}{24} \sigma_{r}\left(\bZ\right) = \frac{1}{24} \sigma_{r}\left(\bZ\bQ\right)$. In addition,  we have 
\begin{align}
\| \bH_{1}\bV^T + \bU\bH_{2}^{T}\|_F & \leq  \sqrt{2}\| \bH\bW^T\|_F \nonumber \\
& \leq \sqrt{2}\| \bH\bH\|_F + \sqrt{2} \|\bH\bQ^T\bZ^T\|_F \nonumber\\ 
&\leq \frac{25}{24} \sqrt{2}\| \bH\bQ^T\bZ^T\|_F,\label{HU_upper}
\end{align}
and
\begin{align} 
\left\Vert \bU\bV^{T} - \bX\bY^{T} \right\Vert_{F}
& \le \frac{1}{\sqrt{2}} \left\Vert \bW\bW^{T} - \bZ\bZ^{T} \right\Vert_{F} \nonumber\\
& = \frac{1}{\sqrt{2}}  \left\Vert \bH\bH^{T} + \bZ\bQ\bH^{T} + \bH\left(\bZ\bQ\right)^{T} \right\Vert_{F}\nonumber \\
 &\le \frac{1}{\sqrt{2}}  \left\Vert \bH\bH^{T} \right\Vert_{F} + \sqrt{2}\left\Vert \bH\bQ^{T}\bZ^T \right\Vert_{F}\nonumber \\
& \leq\frac{49}{48}\sqrt{2}\left\Vert \bH\bQ^{T}\bZ^T \right\Vert_{F},\label{difference_upper}
\end{align}
and 
\begin{equation}\label{eq_bound_h12}
\left\Vert \bH_{1}\bH_{2}^{T}\right\Vert_{F} \le \frac{1}{\sqrt{2}}\left\Vert \bH\bH^{T}\right\Vert_{F} .
\end{equation}

Plugging \eqref{RC_lhs_term1}, \eqref{HU_upper}, \eqref{difference_upper} and \eqref{eq_bound_h12} into \eqref{RC_lhs}, we have
\begin{align}
& \langle \nabla h\left(\bW\right), \bW - \bZ\bQ \rangle  \nonumber \\
& \geq \left[\frac{7}{8}\lambda - (0.0006+0.71\sqrt{s})\alpha_h\frac{25\cdot 49}{24^2}- 0.36\alpha_h\frac{49}{2\cdot24^2} \right]\| \bH\bQ^T\bZ^T\|_F^2\nonumber\\
&\quad  + \lambda \|\bH^T\bZ\|_F^2 + \lambda\left\Vert \bH^{T}\bD\bZ \right\Vert_{F}^{2} +  \lambda \|\bH\bH^{T}\|_F^2 \nonumber \\
&\geq \left(0.1188 - 9.06 \sqrt{s}\right)\| \bH\bQ^T\bZ^T\|_F^2 + \lambda \|\bH^T\bZ\|_F^2 + \lambda\left\Vert \bH^{T}\bD\bZ \right\Vert_{F}^{2} +  \lambda \|\bH\bH^{T}\|_F^2, \label{RC_lhs_sorted}
\end{align}
where \eqref{RC_lhs_sorted} follows from the setting $\alpha_{h} = 6$ and $\lambda = \gamma_{1}/4$.
 
On the other end, since 
\begin{align*}
\left\Vert  \bH^{T}\bD\bH \right\Vert_{F}^{2} 
& = \left\Vert  \bH_{1}^{T}\bH_{1} - \bH_{2}^{T}\bH_{2} \right\Vert_{F}^{2} \\
& \le 2\left(\left\Vert  \bH_{1}\bH_{1}^{T} \right\Vert_{F}^{2} + \left\Vert \bH_{2}\bH_{2}^{T} \right\Vert_{F}^{2}\right)\\
& \le 2\left\Vert  \bH\bH^{T} \right\Vert_{F}^{2} ,
\end{align*}
from Proposition~\ref{prop_local_smooth_clean} and \eqref{equ_regu_g_local_smooth} we have
\begin{align*}
\left\Vert \nabla h\left(\bW\right)  \right\Vert_{F}^{2}
& \le 2\left\Vert \nabla f_{tr}\left(\bW\right)  \right\Vert_{F}^{2}  +  2\left\Vert \nabla g\left(\bW\right)  \right\Vert_{F}^{2}\\
& \le 0.5 \alpha_{h}^{2}  \left\Vert \bU\bV^{T} - \bX\bY^{T} \right\Vert_{F}^{2} \left\Vert \bW \right\Vert^{2} + 2\lambda^{2} \left\Vert  \bW \right\Vert^{2} \left(\left\Vert  \bH^{T}\bD\bH \right\Vert_{F}  + 2\left\Vert  \bH^{T}\bD\bZ \right\Vert_{F} \right)^{2}  \\
& \le \left( 0.5 \alpha_{h}^{2}  \left\Vert \bU\bV^{T} - \bX\bY^{T} \right\Vert_{F}^{2}  + 4\lambda^{2} \left\Vert  \bH^{T}\bD\bH \right\Vert_{F}^{2} + 16\lambda^{2} \left\Vert  \bH^{T}\bD\bZ \right\Vert_{F}^{2} \right)\left\Vert \bW \right\Vert^{2} \\
& \le \left( 0.5 \alpha_{h}^{2} 2\left(\frac{49}{48}\right)^{2}\left\Vert \bH\bQ^{T}\bZ^T \right\Vert_{F}^{2}  + 8\lambda^{2} \left\Vert  \bH\bH^{T} \right\Vert_{F}^{2} + 16\lambda^{2} \left\Vert  \bH^{T}\bD\bZ \right\Vert_{F}^{2} \right) \left(\frac{25}{24}\right)^{2} \left\Vert \bZ \right\Vert^{2} \\
& \le \left(  40.8 \left\Vert \bH\bQ^{T}\bZ^T \right\Vert_{F}^{2}  +  1.1 \left\Vert  \bH^{T}\bD\bZ \right\Vert_{F}^{2} \right) \left\Vert \bZ \right\Vert^{2}  .
\end{align*}
Therefore, if we let $\alpha = 20$ and $\beta = 1000$, we have the right hand side of RC as
\begin{align*}
 \frac{\sigma_{r}^{2}\left(\bZ\right)}{\alpha} \left\Vert \bH \right\Vert_{F}^{2}  + \frac{1}{\beta\left\Vert\bZ\right\Vert^{2}} \left\Vert \nabla h\left(\bW\right)  \right\Vert_{F}^{2}
&\le \frac{\sigma_{r}^{2}\left(\bZ\right)}{20} \left\Vert\bH\right\Vert_{F}^{2} + 0.0408 \left\Vert \bH\bQ^{T}\bZ^T \right\Vert_{F}^2 + 0.0011 \left\Vert  \bH^{T}\bD\bZ \right\Vert_{F}^{2} \\
& \leq 0.0908 \left\Vert \bH\bQ^{T}\bZ^T \right\Vert_{F}^2 + 0.0011 \left\Vert  \bH^{T}\bD\bZ \right\Vert_{F}^{2}.
\end{align*}
Consequently, matching it with the \eqref{RC_lhs_sorted}, we conclude that when $s$ is a sufficiently small constant, RC holds with parameters $(20, 100, \sigma_r(\bZ)/24)$. Note that the parameters $\alpha,\beta, s$ have not been optimized in the proof.

\section{Proof of Robust Initialization} \label{sec_initialization}

As in the description of Algorithm~\ref{alg_minimize_quad_residual}, we split the samples into two portions $\{\by_1,\by_2\}$ in the initialization stage for the convenience of theoretical analysis.  We use the measurements $\by_{2} = \{y_{i}\}_{i=m_{1}+1}^{m}$ to estimate $\left\Vert \bM \right\Vert_{F}$ via the sample median of $\by_2$. Then, we employ the rest of measurements $\by_{1} = \{y_{i}\}_{i=1}^{m_{1}}$ to generate initialization via the truncated spectral method. Besides, denote the outlier fraction of $\by_{1}$ and $\by_{2}$ by $s_{1} = \left\vert\mathcal{S}_{1}\right\vert/m_{1}$ and $s_{2} = \left\vert\mathcal{S}_{2}\right\vert/m_{2}$, respectively, where $\mathcal{S}_{1}$ and $\mathcal{S}_{2}$ are the corresponding outlier supports of $\by_1$ and $\by_2$. Hence, $\max\{s_1,s_2\}\leq 2s$.


Due to Lemma~\ref{lemma_samp_quantile_bound_outlier}, provided $s_{2}$ is small, we have
\begin{equation}
\theta_{\frac{1}{2} - s_{2}}\left( \left\{ \left\vert \mathcal{A}_i(\bM) \right\vert \right\}_{i=m_{1}+1}^{m} \right)
\le \mathrm{med}\left(\left\vert \by_{2} \right\vert\right)
\le \theta_{\frac{1}{2} + s_{2}}\left( \left\{ \left\vert \mathcal{A}_i(\bM) \right\vert \right\}_{i=m_{1}+1}^{m} \right).
\end{equation}


Following Proposition~\ref{prop_sample_median_concen}, if $s_2\leq 2s<0.01$, we have that provided $m\geq c_1nr\log n$ for some large constant $c_1$,
\begin{equation}  \label{threshold_bound_noisy}
0.65 \left\Vert \bM \right\Vert_{F} \le \mathrm{med} \left( \left\vert \by_{2} \right\vert \right) \le 0.70 \left\Vert \bM \right\Vert_{F}
\end{equation}
holds with probability at least $1 - \exp{\left( -c_2m \right)}$ for some constant $c_2$.

Therefore, \eqref{threshold_bound_noisy} guarantees that the threshold used in the truncation is on the order of $\|\bM\|_F$. To emphasize the independence between the measurements used for norm estimation via the sample median and the rest of the measurements used in the truncated spectral method, we define $C_{M} :=  \mathrm{med}\left( \left\vert \by_{2} \right\vert \right)$, which satisfies \eqref{threshold_bound_noisy}. Rewrite \eqref{equ_initial_matrix} as 
\begin{align*}
\bK & = (1-s_1)\bK_1 +s_1 \bK_2 
\end{align*}
where
\begin{equation}\label{def_Y1Y2_noisy}
\bK_1 = \frac{1}{\left\vert \mathcal{S}_{1}^{c} \right\vert} \sum_{i\in \mathcal{S}_{1}^{c}} \mathcal{A}_i(\bM)  \bA_i \mathbb{I}_{\left\{ \left\vert \mathcal{A}_i(\bM)  \right\vert \le   \alpha_{y} C_{M}   \right\}},\quad  
\bK_2 = \frac{1}{\left\vert \mathcal{S}_{1} \right\vert} \sum_{i\in \mathcal{S}_{1}} y_i \bA_i \mathbb{I}_{\left\{ \left\vert y_i \right\vert  \le  \alpha_{y} C_{M} \right\}},
\end{equation}
where $\mathcal{S}_{1}^{c}$ is the the complementary set of $\mathcal{S}_{1}$. Note that
\begin{align*}
\mathbb{E}[\bK_1]=  \frac{1}{\left\vert \mathcal{S}_{1}^{c} \right\vert} \sum_{i\in \mathcal{S}_{1}^{c}} \mathbb{E}\left[\left( \mathcal{A}_i(\bM) \right)  \bA_i \mathbb{I}_{\left\{ \left\vert \mathcal{A}_i(\bM) \right\vert \le   \alpha_{y} C_{M}    \right\}} \right] =  \gamma_{2}   \bM,
\end{align*}
where $\gamma_{2} := \mathbb{E} \left[   \xi^{2}   \mathbb{I}_{\left\{  \left\vert \xi \right\vert \le \alpha_{y} C_{M} /  \left\Vert \bM \right\Vert_{F} \right\}} \right] \le 1$ with $\xi \sim \mathcal{N}\left(0, 1\right)$, and 
\begin{equation*}
\mathbb{E}[\bK_2] =  \frac{1}{\left\vert \mathcal{S}_{1} \right\vert} \sum_{i\in \mathcal{S}_{1}} y_i \mathbb{E}[\bA_i] \mathbb{I}_{\left\{ \left\vert y_i \right\vert  \le  \alpha_{y} C_{M} \right\}} =  \boldsymbol{0}.
\end{equation*}
We have the following proposition on the concentration of $\bK$, of which the proof is given in Appendix~\ref{proof_bound_Y_noisy}.

\begin{proposition}\label{bound_Y_noisy}
With probability at least $1-n^{-c_1}$, we have
\begin{equation}\label{concentration_Y}
\left\Vert \bK - \left(1-s_1\right)  \gamma_{2}   \bM \right\Vert \le C \alpha_y   \sqrt{\frac{n\log n}{m}} \|\bM\|_F,
\end{equation}
provided that $m\geq c_2 \log n$, where $c_1,c_2,C>1$ are numerical constants.
\end{proposition}


Let $\epsilon:=C \alpha_y   \sqrt{\frac{n\log n}{m}}$ for short-hand notations. Denote $\tilde{n} = \min\{n_{1}, n_{2}\}$. Let $\sigma_1(\bK)\geq \sigma_2(\bK)\geq \cdots \sigma_{\tilde{n}}(\bK)$ be the singular values of $\bK$ in a nonincreasing order, and $\sigma_1(\bM)\geq \sigma_2(\bM)\geq \cdots \sigma_{\tilde{n}}(\bM)$ be the singular values of $\bM$ in a nonincreasing order. Since $\bM$ has rank $r$, we know $\sigma_{r+1}(\bM)  = \dots = \sigma_{\tilde{n}}(\bM) = 0$. By the Weyl's inequality and \eqref{concentration_Y}, we have 
\begin{equation}
\left\vert \sigma_i(\bK) - \left(1-s_1\right) \gamma_{2}\sigma_{i}(\bM) \right\vert \le \epsilon \left\Vert \bM \right\Vert_F, \quad i=1,2,\dots,\tilde{n},
\end{equation}
which implies 
\begin{equation}
\sigma_{i}(\bK)  \le \epsilon \left\Vert \bM \right\Vert_F, \quad i\ge(r+1).
\end{equation}


By definition, $\bU_{0} = \bC_{L} \boldsymbol{\Sigma}^{1/2} $, $\bV_{0} = \bC_{R} \boldsymbol{\Sigma}^{1/2} $ and $\bW_{0} = \begin{bmatrix} \bU_{0} \\ \bV_{0} \end{bmatrix}$, where $\bC_{L}\boldsymbol{\Sigma}\bC_{R}^{T}:=$ rank-$r$ SVD of $\bK$, with $\bC_{L} \in \mathbb{R}^{n_{1}\times r}$, $\bC_{R} \in \mathbb{R}^{n_{2}\times r}$ and $\boldsymbol{\Sigma} \in \mathbb{R}^{r\times r}$. Recalling $\bZ = \begin{bmatrix} \bX \\ \bY \end{bmatrix}$, then according to Lemma~\ref{lemma_fronorm_enhancebound}, we have
\begin{align*}
\left\Vert \bW_{0}\bW_{0}^{T} -  \bZ\bZ^{T} \right\Vert_{F}
& \le 2 \left\Vert \bU_{0}\bV_{0}^{T} - \bM \right\Vert_{F} \\
& = 2 \left\Vert \bC_{L}\boldsymbol{\Sigma}\bC_{R}^{T} - (1-s_1)\gamma_2\bM\right\Vert_F + 2\left\Vert \left((1-s_1)\gamma_2 - 1\right)\bM \right\Vert_{F}  \\
& \leq 2\sqrt{2r}\left(  \left\Vert \bC_{L}\boldsymbol{\Sigma}\bC_{R}^{T} - \bK \right\Vert+\left\Vert \bK - (1-s_1)\gamma_2 \bM \right\Vert \right) + 2 \left| (1-s_1)\gamma_2 - 1\right| \cdot \| \bM\|_F  \\
& \leq  2\sqrt{2r} \left( \sigma_{r+1}(\bK) + \epsilon \|\bM\|_F \right) + 2\left| (1-s_1)\gamma_2-1\right| \cdot \| \bM\|_F \\
& \le \left( 4\sqrt{2r} \epsilon +  2s_1\gamma_2 +  2(1-\gamma_2)  \right) \|\bM\|_F .
\end{align*}
By Lemma~\ref{lemma_fronorm_lowbound}, we have
\begin{align*}
\mathrm{dist}\left(\bW_{0}, \bZ\right)
&\le \frac{\left\Vert \bW_{0}\bW_{0}^{T} -  \bZ\bZ^{T} \right\Vert_{F}}{\sqrt{2\left(\sqrt{2}-1\right)} \sigma_{r} \left( \bZ \right)} \\
&\le \frac{ \left( 4\sqrt{2r} \epsilon +  2s_1\gamma_2 +  2(1-\gamma_2)  \right) \|\bM\|_F}{\sqrt{2\left(\sqrt{2}-1\right)} \sigma_{r} \left( \bZ \right)} \\
& = \frac{ \left( 2\sqrt{2r} \epsilon +  s_1\gamma_2 +  (1-\gamma_2)  \right) \|\bM\|_F}{ \sqrt{\sqrt{2}-1} \sqrt{\sigma_{r} \left( \bM \right)}},
\end{align*} 
where we use the fact that for all $i$, $\sigma_{i}(\bX) = \sigma_{i}(\bY) = \sigma_{i}(\bZ)/\sqrt{2} = \sqrt{\sigma_{i} \left( \bM \right)}$.

Therefore, we have $\mathrm{dist}\left(\bW_{0}, \bZ\right) \le \frac{1}{24} \sigma_{r}\left(\bZ\right)$ if 
\begin{equation*}
\max\{\sqrt{r}\epsilon , s_1, 1-\gamma_2 \} \leq c\frac{\sigma_r(\bM)}{\|\bM\|_F} = \frac{c}{\sqrt{r}\bar{\kappa} } .
\end{equation*}
To be more specific, we need $s_1<2s\leq   c_1/(\sqrt{r}\bar{\kappa})$, $m> c_2 \alpha_y^2 nr^2 \bar{\kappa}^2\log n $, and 
\begin{align*}
1-\gamma_2  =  \mathbb{E}_{\xi\sim\mathcal{N}(0, 1)}\left[ \xi^{2}  \mathbb{I}_{\left\{  \left\vert \xi  \right\vert > \alpha_{y} C_{M} /  \left\Vert \bM \right\Vert_{F}\right\} } \right] \leq \frac{1}{35\sqrt{r}\bar{\kappa}}.
\end{align*}
The last condition can be satisfied by setting $\alpha_y = 2\log{(r^{1/4}\bar{\kappa}_{0}^{1/2} + 20)}$, as long as $\bar{\kappa}_0$ is an upper bound of $\bar{\kappa}$ such that $\bar{\kappa}\leq \bar{\kappa}_0$.



\section{Conclusion}\label{sec_conclusions}

In this paper, we present a median-truncated gradient descent algorithm to improve the robustness of low-rank matrix recovery from random linear measurements in the presence of outliers. The effectiveness of the proposed algorithm is provably guaranteed by theoretical analysis, and validated through various numerical experiments as well. In the future work, we will extend the proposed approach to robust low-rank recovery problems such as robust PCA and blind deconvolution.

\section*{Acknowledgements}
Preliminary results of this paper were presented in part at 2017 International Conference on Sampling Theory and Applications (SampTA) \cite{li2017nonconvex}. The work of Y. Li and Y. Chi is supported in part by AFOSR under the grant FA9550-15-1-0205, by ONR under the grant N00014-15-1-2387, and by NSF under the grants CAREER ECCS-1650449, ECCS-1462191 and CCF-1704245. The work of H. Zhang and Y. Liang is supported in part by AFOSR under the grant AFOSR FA9550-16-1-0077, by NSF under grants ECCS 16-09916 and CCF 17-04169.


\appendix
\section*{Appendices}

\section{Useful Lemmas}

\begin{lemma}\cite[Lemma 1]{zhang2016provable}\label{lemma_quantile_func_bound}
Suppose $F(\cdot)$ is cumulative distribution function (i.e., non-decreasing and right-continuous) with continuous density function $f(\cdot)$. Assume the samples $\{X_{i}\}_{i=1}^{m}$ are i.i.d. drawn from $f$. Let $0<p<1$. If $l < f(\theta) < L$ for all $\theta$ in $\{\theta: \left\vert \theta - \theta_{p}\right\vert \le \epsilon \}$, then 
\begin{equation*}
\left\vert \theta_{p}(\{X_{i}\}_{i=1}^{m}) - \theta_{p}(F)\right\vert < \epsilon
\end{equation*}
holds with probability at least $1-2\exp{(-2m\epsilon^{2}l^{2})}$.
\end{lemma}

\begin{lemma}\cite[Lemma 2]{zhang2016provable}\label{lem_orderstats}
Given a vector $\bX=(X_1,X_2,...,X_n)$, where we order the entries in a non-decreasing manner $X_{(1)}\le X_{(2)}\le ... \le X_{(n-1)}\le X_{(n)}$.
Given another vector $\bY=(Y_1,Y_2,..., Y_n)$, then 
\begin{flalign}
	|X_{(k)}-Y_{(k)}|\le \|\bX-\bY\|_\infty,
\end{flalign}
holds for all $k=1,..., n.$
\end{lemma}

\begin{lemma}\cite[Lemma 3]{zhang2016provable}\label{lemma_samp_quantile_bound_outlier}
Consider clean samples $\{\tilde{y}_{i}\}_{i=1}^{m}$. If a fraction $s$ of them are corrupted by outliers, one obtains contaminated samples $\{y_{i}\}_{i=1}^{m}$, which contain $sm$ corrupted samples and $(1-s)m$ clean samples. Then for a quantile $p$ such that $s < p < 1-s$, we have 
\begin{equation}
\theta_{p-s}\left( \left\{\tilde{y}_{i}\right\}_{i=1}^{m} \right)  \le  \theta_{p}\left( \left\{y_{i}\right\}_{i=1}^{m} \right) \le \theta_{p+s}\left( \left\{\tilde{y}_{i}\right\}_{i=1}^{m} \right).
\end{equation}
\end{lemma}

\begin{lemma}\cite[Lemma 5.4]{tu2016low}\label{lemma_fronorm_lowbound}
For any $\bX$, $\bU \in \mathbb{R}^{n \times r}$, we have 
\begin{equation}\label{dist_upper}
\left\Vert \bX\bX^{T} - \bU\bU^{T} \right\Vert_{F} \ge \sqrt{2\left(\sqrt{2}-1\right)} \sigma_{r} \left( \bX \right) \mathrm{dist}(\bX,\bU).
\end{equation}
\end{lemma}

\begin{lemma}\cite[Lemma 4]{zheng2016convergence}\label{lemma_fronorm_enhancebound}
For any matrix $\bZ_{i}$ of the form $\bZ_{i} = \begin{bmatrix} \bU_{i}\boldsymbol{\Sigma}_{i}^{\frac{1}{2}}\bQ_{i} \\  \bV_{i}\boldsymbol{\Sigma}_{i}^{\frac{1}{2}}\bQ_{i} \end{bmatrix}$, where $\bU_{i}$, $\bV_{i}$ and $\bQ_{i}$ are unitary matrices and $\boldsymbol{\Sigma}_{i} \succeq 0$ is a diagonal matrix, for $i=1,2$, we have 
\begin{equation}
\left\Vert \bZ_{1}\bZ_{1}^{T} - \bZ_{2}\bZ_{2}^{T} \right\Vert_{F} \le 2\left\Vert \bU_{1}\boldsymbol{\Sigma}_{1}\bV_{1}^{T} - \bU_{2}\boldsymbol{\Sigma}_{2}\bV_{2}^{T} \right\Vert_{F}.
\end{equation}
\end{lemma}

\begin{lemma}[Orlicz-norm version Bernstein's inequality]\cite[Proposition 2]{koltchinskii2011nuclear}\label{lemma_orlicz_matrix_bernstein}
Let $\bS_{1}, \bS_{2}, \dots, \bS_{m}$ be a finite sequence of independent zero-mean random matrices with dimensions $d_{1}\times d_{2}$. Suppose $\left\Vert \bS_{i} \right\Vert_{\psi_{2}} \le B$, and define $\sigma_{\bS}^{2} = \max\left\{ \left\Vert \frac{1}{m} \sum_{i=1}^{m}  \mathbb{E}\left[ \bS_{i}\bS_{i}^{T} \right]  \right\Vert, \left\Vert \frac{1}{m} \sum_{i=1}^{m}  \mathbb{E}\left[ \bS_{i}^{T}\bS_{i} \right]  \right\Vert \right\}$. Then there exists a constant $C>0$ such that, for all $t>0$, with probability at least $1- e^{-t}$
\begin{equation*}
\left\Vert \frac{1}{m}\sum_{i=1}^{m} \bS_{i} \right\Vert  \le C \max\left\{ \sigma_{\bS}\sqrt{\frac{t + \log{\left(d_{1}+d_{2}\right)}}{m}},  B \sqrt{\log{\left(\frac{B}{\sigma_{\bS}}\right)}} \frac{t+\log{\left(d_{1}+d_{2}\right)}}{m} \right\}.
\end{equation*}
\end{lemma}

\begin{lemma}[Covering number for low-rank matrices]\label{lemma_covering_net} 
Let $S_r = \{\bX\in\mathbb{R}^{n_1\times n_2},\mbox{rank}(\boldsymbol{X})\leq r, \| \boldsymbol{X}\|_\textrm{F}=1\}$. Then there exists an $\epsilon$-net $\bar{S}_r\subset S_r$ with respect to the Frobenius norm obeying
$$ |\bar{S}_r| \leq (9/\epsilon)^{(n_1+n_2+1)r}. $$
\end{lemma}

\begin{lemma}\label{lemma_max_fnorm}
Suppose $\bA_{i} \in \mathbb{R}^{n_{1}\times n_{2}}$'s are sensing matrices, each generated with i.i.d. Gaussian entries, for $i = 1,2,\dots,m$. Let $n = (n_1+n_2)/2$, and $m\geq n$. Then 
\begin{align}\label{union_max_fnorm}
\max_{i=1,2,\dots,m} \left\Vert \bA_{i} \right\Vert_{F} \le 2 \sqrt{n \left(n+m\right)}
\end{align}
holds with probability exceeding $1 - m\cdot \exp{\left(- n\left(n+m\right) \right)}$.
\end{lemma}

\begin{proof}
Let $\bA$ be a sensing matrix, generated with i.i.d. standard Gaussian entries, and $\bA_{k,t}$ be the entry of $\bA$ with index $\left(k, t\right)$, then we know $\bA_{k,t} \sim \mathcal{N}\left(0,1\right)$. Since $\left\Vert \bA \right\Vert_{F}^{2} = \sum_{k,t} \bA_{k,t}^{2}$, $\left\Vert \bA \right\Vert_{F}^{2}$ is a Chi-squared random variable with degree of freedom as $n_{1}n_{2}$. According to \cite[Lemma 1]{laurent2000adaptive}, we have 
\begin{align*}
\mathbb{P} \left\{ \left\Vert \bA \right\Vert_{F}^{2} \ge \left(1 + 2 \sqrt{\lambda}  + 2\lambda\right) n_{1}n_{2} \right\} \le \exp{\left(-\lambda n_{1}n_{2} \right)},
\end{align*}
for any $\lambda > 0$. Let $\lambda = \left(n+m\right)/n$. It is clear that $\lambda \ge 2$ for $m\ge n$. Moreover, $2\lambda \ge 2\sqrt{\lambda} + 1$ for $\lambda \ge 2$. Thus, we obtain
\begin{align*}
\mathbb{P} \left\{ \left\Vert \bA \right\Vert_{F}^{2} \ge 4 n\left(n+m\right) \right\} \le \exp{\left(- n\left(n+m\right)\right)}.
\end{align*}
Therefore the proof is completed by applying the union bound.
\end{proof}


\begin{lemma}[Restricted Isometry Property]\label{lemma_RIP_A_rec}
Fix $0 < \delta < 1$. For every $1\le r\le \min\{n_{1}, n_{2}\}$, there exist positive constants $c_{0}$ and $c_{1}$ depending only on $\delta$ such that provided $m \ge c_{0} (n_{1}+n_{2})r  $, 
\begin{equation}
\left(1-\delta\right) \left\Vert \bM\right\Vert_{F}  \le \frac{1}{\sqrt{m}} \left\Vert \mathcal{A}\left(\bM\right) \right\Vert_{2}  \le  \left(1+\delta\right) \left\Vert \bM\right\Vert_{F}
\end{equation}
holds for all matrices $\bM$ of rank at most $r$ with probability at least $1-\exp{\left( - c_{1} m\right)}$.
\end{lemma}

\section{Proof of Proposition~\ref{prop_sample_median_concen}}\label{proof_prop_sample_median_concen}

Due to scaling invariance, without loss of generality, it is sufficient to consider all rank-$2r$ matrices with unit Frobenius norm. First, we fix the rank-$2r$ matrix $\bG_0\in\mathbb{R}^{n_{1}\times n_{2}}$, and then generalize to all rank-$2r$ matrices by a covering argument. Note that  $\left\vert  \mathcal{A}_{i}\left( \bG_0 \right) \right\vert$, $i = 1,2, \dots, m$, are i.i.d. copies of 
$ \left\vert  \langle \bA, \bG_0  \rangle   \right\vert $, 
where $\bA$ is generated with i.i.d. Gaussian entries.  Since $\left\Vert \bG_0  \right\Vert_{F} = 1$, $\langle \bA, \bG_0 \rangle$ follows the distribution $\mathcal{N}\left(0, 1\right)$, and $ \left\vert \langle \bA, \bG_0 \rangle \right\vert$ follows a folded normal distribution, whose probability density function and cumulative distribution function are denoted by $f_{1}$ and $F_{1}$, respectively. It is known from Lemma~\ref{lemma_quantile_func_bound} that 
\begin{align}\label{equ_med_concentration_fixed}
0.6745 - \epsilon\leq  \mathrm{med}\left(\left\vert\mathcal{A} (\bG_0)   \right\vert  \right)  \leq 0.6745 + \epsilon  , 
\end{align}
with probability at least $1 - 2 \exp{\left( -cm\epsilon^{2} \right)}$ for a small $\epsilon$, where $c$ is a constant around $2\times 0.6356^2$. Similar arguments extend to other quantiles. From Lemma~\ref{lemma_quantile_func_bound}, we have
\begin{align}
0.6588 - \epsilon\leq & \theta_{0.49}\left(\left\vert  \mathcal{A}\left(\bG_0 \right)   \right\vert  \right)  \leq 0.6588 + \epsilon  , \\
0.6903 - \epsilon \leq &\theta_{0.51}\left(\left\vert  \mathcal{A}\left(\bG_0 \right)   \right\vert  \right)  
\leq 0.6903 + \epsilon  , 
\end{align}
with probability at least $1 - 2 \exp{\left( -cm\epsilon^{2} \right)}$ for a small $\epsilon$, where $c$ is a constant around $2\times 0.6287^2$.

Next, we extend the results to all rank-$2r$ matrices $\bG$ with $\|\bG\|_F=1$ via a covering argument. We argue for the median and similar arguments extend to other quantiles straightforwardly. Let $\mathcal{N}_{\tau}$ be a $\tau$-net covering all rank-$2r$ matrices with respect to the Frobenius norm. Let $n = (n_1+n_2)/2$, then from Lemma~\ref{lemma_covering_net}, $\left\vert \mathcal{N}_{\tau} \right\vert \le \left( 9/\tau \right)^{2r\left(2n+1\right)}$. Taking the union bound, we obtain 
\begin{equation}
0.6745 - \epsilon \le  \mathrm{med}\left( \left\vert \mathcal{A}( \bG_{0})   \right\vert  \right)   \le 0.6745 + \epsilon, \quad \forall  \bG_{0} \in \mathcal{N}_{\tau},
\end{equation}
with probability at least $1 - \left( 9/\tau \right)^{2r\left(2n+1\right)}\exp{\left( -cm\epsilon^{2} \right)}$. 
Set $\tau = \epsilon /  (2 \sqrt{ n (n+m )} )$. Under this event and  \eqref{union_max_fnorm}, which holds with probability at least $1 - m\exp{\left(- n\left(n+m\right) \right)}$ from Lemma~\ref{lemma_max_fnorm}, for any rank-$2r$ matrix $\bG$ with $\|\bG\|_F=1$, there exists $\bG_0\in\mathcal{N}_\tau$ such that $\left\Vert \bG - \bG_{0} \right\Vert_{F} \le \tau$, and
\begin{align}
 \left\vert \mathrm{med}\left( \left\vert \mathcal{A}\left( \bG_{0} \right) \right\vert \right) - \mathrm{med}\left( \left\vert \mathcal{A}\left( \bG \right) \right\vert  \right) \right\vert
& \le \max_{i=1,2,\dots,m} \big| \left\vert \langle \bA_{i}, \bG_{0} \rangle \right\vert - \left\vert \langle \bA_{i}, \bG \rangle \right\vert \big| \label{equ_med_max}\\
& \le \max_{i=1,2,\dots,m} \left\vert  \langle \bA_{i}, \bG_{0} \rangle -  \langle \bA_{i},  \bG \rangle  \right\vert \label{equ_absolute_bound}\\
& \le \max_{i=1,2,\dots,m}  \left\Vert \bG_{0} - \bG \right\Vert_{F} \left\Vert \bA_{i} \right\Vert_{F} \nonumber\\
& \le \tau \max_{i=1,2,\dots,m}  \left\Vert \bA_{i} \right\Vert_{F} \leq \epsilon\label{equ_max_fnorm_bound},
\end{align}
where \eqref{equ_med_max} follows from Lemma~\ref{lem_orderstats}, and \eqref{equ_absolute_bound} follows from the fact $\left\vert \left\vert a \right\vert - \left\vert b \right\vert \right\vert \le \left\vert a-b \right\vert$.

The rest of the proof is then to argue that \eqref{equ_max_fnorm_bound} holds with probability at least $1 - c_{1}\exp{\left( -c_{2}m\epsilon^{2} \right)}$ for some constants $c_{1}$ and $c_{2}$, as long as $m \ge c_{0} \left(\epsilon^{-2} \log{\epsilon^{-1}}\right) nr \log{\left(nr\right)}$ for some sufficiently large constant $c_{0}$. Note that 
\begin{align*}
\left( 9/\tau \right)^{2r\left(2n+1\right)} 
& = \exp\left(2r\left(2n+1\right) \left(\log{18} + \log(\epsilon^{-1}) + \frac{1}{2}\log{n} + \frac{1}{2}\log{\left(n+m\right)}  \right) \right)\\
& \le \exp\left(5nr\log{m} + c_{3}nr\log{\epsilon^{-1}}\right).
\end{align*}
It is straightforward to verify $c_{3}nr\log{\epsilon^{-1}} \le c_{4} m \epsilon^{2}$, where $2c_{4} < c - c_{2}$, based on the specific setting of $m$, as long as $c_{0}$ is large enough. Then, it suffices to show 
\begin{equation}\label{equ_number_meas_median_conc}
5nr\log{m} < c_{5} m \epsilon^{2},
\end{equation} 
where $c_{5} < c - c_{4} - c_{2}$, when $m \ge c_{0} \left(\epsilon^{-2} \log{\epsilon^{-1}}\right) nr \log{\left(nr\right)}$ for some large enough constant $c_{0}$.

First, for any fixed $n$, if \eqref{equ_number_meas_median_conc} holds for some $m$ and $m \ge \left(5/c_{5}\right) \epsilon^{-2}nr$, then \eqref{equ_number_meas_median_conc} holds for a larger $m$, since
\begin{align*}
5nr\log{\left(m+1\right)} = 5nr\log{m} + \frac{5nr}{m}\log{\left(1+\frac{1}{m}\right)^{m}} \le 5nr\log{m} + 5nr/m \le c_{5} \left(m + 1\right) \epsilon^{2}. 
\end{align*}

Next, we show that for any fixed $n$, we can find a constant $c_{0}$ such that \eqref{equ_number_meas_median_conc} holds as long as $m = c_{0} \left(\epsilon^{-2} \log{\epsilon^{-1}}\right) nr \log{\left(nr\right)}$. Pick a small enough $\epsilon < 1/e$ that is fixed throughout the proof. Given $c_{5}$, we can always find a large enough $c_{0}$ such that $\frac{1}{3}\log{c_{0}} < c_{5}c_{0}/15 - 5/3$. Then as long as $nr\ge 3$, we can get $\frac{1}{3}\log{c_{0}} < \left(c_{5}c_{0}/15 - 5/3\right) \log{\epsilon^{-1}} \log{nr}$, which further yields $\frac{1}{3}\log{c_{0}} + \log{\epsilon^{-1}} + \frac{2}{3} \log{nr} < \left(c_{5}c_{0}/15\right) \log{\epsilon^{-1}} \log{nr}$. As a result, we have
\begin{align*}
\left(c_{5}c_{0}/5\right) \log{\epsilon^{-1}} \log{nr}
& > \log{c_{0}} + 3\log{\epsilon^{-1}} + 2 \log{nr}\\
& = \log{\left(c_{0} \epsilon^{-3} \left(nr\right)^{2}\right)}\\
& > \log{\left( c_{0} \left(\epsilon^{-2}\log{\epsilon^{-1}}\right) nr \log{\left(nr\right)}  \right)},
\end{align*}
which implies \eqref{equ_number_meas_median_conc}.

\section{Proof of Proposition~\ref{prop_local_curv_clean_first}}\label{proof_prop_local_curv_clean_first}


We prove the following lemma which directly implies Proposition~\ref{prop_local_curv_clean_first}.
\begin{lemma}\label{lemma_prop2}
Under the conditions of Proposition~\ref{prop_local_curv_clean_first}, we have
\begin{equation*}
\frac{1}{m}\sum_{i=1}^{m}  \langle \bA_{i},\bG \rangle  \cdot \langle \bA_{i}, \bT \rangle  \cdot    \mathbb{I}_{ \left\{   \left\vert  \langle \bA_{i}, \bG \rangle  \right\vert   \le   0.65 \alpha_{h}  \left\Vert \bG \right\Vert_{F}  \right\}  } \mathbb{I}_{ \left\{ \langle \bA_{i}, \bG  \rangle  \cdot \langle \bA_{i}, \bT \rangle  \ge 0 \right\} }  \geq \gamma_1 \langle \bG,\bT\rangle - 0.0011\alpha_h \|\bG\|_F\|\bT\|_F
\end{equation*}
holds with high probability for all rank-$2r$ matrices $\bG,\bT\in\mathbb{R}^{n_{1}\times n_{2}}$.
\end{lemma}
Specializing Lemma~\ref{lemma_prop2} to $\bG = \bU\bV^T-\bX\bY^T$ and $\bT=\bH_{1}\bV^T + \bU\bH_{2}^{T}$ yields Proposition~\ref{prop_local_curv_clean_first}. The rest of the proof is dedicated to proving Lemma~\ref{lemma_prop2}. Without loss of generality, we can assume $\|\bG\|_F=\|\bT\|_F=1$. Define an auxiliary function as
\begin{equation*}
\chi\left(t\right) = 
\begin{cases}
1, &\quad \left\vert t\right\vert < 0.65 \alpha_{h} - \delta;\\
\frac{1}{\delta}\left( 0.65 \alpha_{h} - \left\vert t\right\vert\right), & \quad  0.65 \alpha_{h} - \delta \le \left\vert t\right\vert \le 0.65 \alpha_{h};\\
0, & \quad \left\vert t\right\vert > 0.65 \alpha_{h},
\end{cases}
\end{equation*}
where $\delta$ is a sufficiently small constant. The function $\chi\left(t\right)$ is a Lipschitz function with the Lipschitz constant $1/\delta$. We have
\begin{align}
 &\quad \langle \bA_{i}, \bG \rangle \cdot \langle \bA_{i}, \bT \rangle \cdot  \mathbb{I}_{ \left\{   \left\vert  \langle \bA_{i}, \bG \rangle   \right\vert   \le  0.65 \alpha_{h} - \delta \right\}  }  \cdot \mathbb{I}_{ \left\{ \langle \bA_{i}, \bG \rangle  \cdot \langle \bA_{i}, \bT \rangle  \ge 0 \right\} } \nonumber \\
 & \le \langle \bA_{i}, \bG \rangle \cdot \langle \bA_{i}, \bT \rangle \cdot \chi ( \langle \bA_{i}, \bG \rangle  )  \cdot \mathbb{I}_{ \left\{ \langle \bA_{i}, \bG \rangle  \cdot \langle \bA_{i}, \bT \rangle  \ge 0 \right\} }\label{bound_auxiliary}\\
&\le   \langle \bA_{i}, \bG \rangle \cdot \langle \bA_{i}, \bT \rangle  \cdot    \mathbb{I}_{ \left\{   \left\vert  \langle \bA_{i}, \bG \rangle   \right\vert   \le   0.65 \alpha_{h} \right\}  } \cdot \mathbb{I}_{ \left\{ \langle \bA_{i},\bG   \rangle  \cdot \langle \bA_{i}, \bT \rangle  \ge 0 \right\} }. \nonumber
\end{align}

Let $\zeta_{i} = \langle \bA_{i}, \bG \rangle \cdot \langle \bA_{i}, \bT \rangle \cdot \chi ( \langle \bA_{i}, \bG \rangle  )  \cdot \mathbb{I}_{ \left\{ \langle \bA_{i}, \bG \rangle  \cdot \langle \bA_{i}, \bT \rangle  \ge 0 \right\} }$, $i=1,2,\dots, m$, of which each can be considered as an i.i.d. copy of $\zeta$, defined as $\zeta  = \langle \bA, \bG \rangle \cdot \langle \bA, \bT \rangle \cdot \chi ( \langle \bA , \bG \rangle  )  \cdot \mathbb{I}_{ \left\{ \langle \bA, \bG \rangle  \cdot \langle \bA, \bT \rangle  \ge 0 \right\} }$. From \eqref{bound_auxiliary}, we have 
\begin{align*}
\mathbb{E}\left[ \zeta \right] 
& \ge \mathbb{E}\left[ \langle \bA , \bG \rangle \cdot \langle \bA , \bT \rangle \cdot  \mathbb{I}_{ \left\{   \left\vert  \langle \bA , \bG \rangle   \right\vert   \le  0.65 \alpha_{h} - \delta \right\}  }  \cdot \mathbb{I}_{ \left\{ \langle \bA, \bG \rangle  \cdot \langle \bA, \bT \rangle  \ge 0 \right\} } \right] \\
& \ge  \mathbb{E}\left[ \langle \bA , \bG \rangle \cdot \langle \bA , \bT \rangle \cdot  \mathbb{I}_{ \left\{   \left\vert  \langle \bA , \bG \rangle   \right\vert   \le  0.65 \alpha_{h} - \delta \right\}  }    \right]  \\
& =\left\langle \mathbb{E}\left[ \langle \bA , \bG \rangle  \bA   \cdot  \mathbb{I}_{ \left\{   \left\vert   \langle \bA , \bG \rangle    \right\vert   \le  0.65 \alpha_{h} - \delta \right\}   } \right]  , \bT \right\rangle   =  \gamma_{1} \cdot \langle \bG, \bT \rangle,  
\end{align*}
where $\gamma_{1} = \mathbb{E}\left[ \xi^{2}  \mathbb{I}_{\left\{  \left\vert \xi \right\vert \le 0.65 \alpha_{h} - \delta \right\}} \right] $ with $\xi \sim \mathcal{N}\left(0, 1\right)$. Moreover, for $p\geq 0$,
\begin{align*}
\left(\mathbb{E}\left[ \left\vert \zeta \right\vert^{p} \right]\right)^{1/p}
& \le \left(\mathbb{E}\left[ \left\vert  \langle \bA, \bG \rangle \cdot \langle \bA, \bT \rangle \cdot  \mathbb{I}_{ \left\{   \left\vert  \langle \bA, \bG \rangle   \right\vert   \le   0.65  \alpha_{h}    \right\}  }  \cdot \mathbb{I}_{ \left\{ \langle \bA, \bG   \rangle  \cdot \langle \bA, \bT \rangle  \ge 0 \right\} }  \right\vert^{p} \right] \right)^{1/p}\\
& \le  \left(\mathbb{E}\left[ \left\vert  \langle \bA, \bG \rangle \cdot \langle \bA, \bT \rangle \cdot  \mathbb{I}_{ \left\{   \left\vert  \langle \bA, \bG \rangle   \right\vert   \le   0.65  \alpha_{h}    \right\}  }    \right\vert^{p} \right] \right)^{1/p}\\
& \le 0.65  \alpha_{h}  \left(\mathbb{E}\left[ \left\vert \langle \bA, \bT \rangle  \right\vert^{p} \right]\right)^{1/p}\leq 0.65c\alpha_h \sqrt{p},
\end{align*}
which indicates that $\zeta$ is a sub-Gaussian random variable with 
$\left\Vert \zeta \right\Vert_{\psi_{2}} \leq 0.65c \alpha_{h} $. Then applying the Hoeffding-type inequality \cite[Proposition 5.10]{Vershynin2012}, we have for any $t \ge 0$, 
\begin{equation*}
 \mathbb{P}\left\{ \left\vert \frac{1}{m}\sum_{i=1}^{m} \zeta_{i} - \mathbb{E}\left[\zeta\right] \right\vert \ge t \right\}  \le \exp{\left(- cmt^{2}/\alpha_{h}^{2}  \right)}.
\end{equation*}
for some $c>0$. Let $t = \varepsilon \alpha_{h}$, where $\varepsilon$ is small enough. Then 
\begin{align}
\frac{1}{m}\sum_{i=1}^{m} \zeta_{i} 
& \ge \mathbb{E}\left[\zeta\right]  -    \varepsilon \alpha_{h}   \ge \gamma_ 1  \langle\bG,  \bT \rangle  -    \varepsilon \alpha_{h}  \label{eq_local_curv_first_term_fixed}
\end{align}
holds with probability at least $1 - \exp{\left( - c m \varepsilon^{2}  \right)}$.

Next, a covering argument is needed to extend  \eqref{eq_local_curv_first_term_fixed} to all rank-$2r$ matrices $(\bG,\bT)$ with unit Frobenius norm. Let $\mathcal{N}_{\tau}$ be a $\tau$-net covering all rank-$2r$ matrices with respect to the Frobenius norm, and define
\begin{equation*}
\mathcal{M}_{\tau} = \left\{ \left(\bG_{0}, \bT_{0}\right): \left(\bG_{0}, \bT_{0}\right) \in \mathcal{N}_{\tau} \times \mathcal{N}_{\tau}\right\}
\end{equation*}
such that for any pair of rank-$2r$ matrices $\left(\bG, \bT\right)$ with $\left\Vert \bG \right\Vert_{F} = \left\Vert \bT \right\Vert_{F} = 1$, there exists $\left(\bG_{0}, \bT_{0}\right) \in \mathcal{M}_{\tau}$ with $\left\Vert \bG_{0} \right\Vert_{F} = \left\Vert \bT_{0} \right\Vert_{F} = 1$ satisfying $\left\Vert \bG_{0} - \bG \right\Vert_{F} \le \tau$ and $\left\Vert \bT_{0} - \bT \right\Vert_{F} \le \tau$. Since both $\mathrm{rank}\left(\bG\right) \le 2r$ and $\mathrm{rank}\left(\bT\right) \le 2r$, then Lemma~\ref{lemma_covering_net} guarantees $\left\vert \mathcal{M}_{\tau} \right\vert \le \left( 9/\tau \right)^{2r\left(2n+1\right)} \cdot \left( 9/\tau \right)^{2r\left(2n+1\right)} \le \left( 9/\tau \right)^{4r\left(2n+1\right)}$. Taking the union bound gives for all $(\bG_0,\bT_0)\in\mathcal{M}_{\tau}$,
\begin{equation*}
\frac{1}{m}\sum_{i=1}^{m}   \langle \bA_{i}, \bG_0 \rangle \cdot \langle \bA_{i}, \bT_0 \rangle \cdot  \chi \left(    \langle \bA_{i}, \bG_0 \rangle \right)  \cdot \mathbb{I}_{ \left\{ \langle \bA_{i}, \bG_0   \rangle  \cdot \langle \bA_{i}, \bT_0 \rangle  \ge 0 \right\} }  \ge \gamma_{1} \cdot  \langle \bG_0 ,  \bT_0 \rangle  -    \varepsilon \alpha_{h} 
\end{equation*}
with probability at least $1- \left( 9/\tau \right)^{4r\left(2n+1\right)}\exp{\left( - c \varepsilon^{2} m \right)}$. Furthermore, 
\begin{align}
& \quad \Big\vert \frac{1}{m}\sum_{i=1}^{m}   \langle \bA_{i}, \bG \rangle \cdot \langle \bA_{i}, \bT \rangle \cdot  \chi \left(    \langle \bA_{i}, \bG \rangle \right)  \cdot \mathbb{I}_{ \left\{ \langle \bA_{i}, \bG   \rangle  \cdot \langle \bA_{i}, \bT \rangle  \ge 0 \right\} } \nonumber\\
& \quad\quad- \frac{1}{m}\sum_{i=1}^{m}  \langle \bA_{i}, \bG_{0} \rangle \cdot \langle \bA_{i}, \bT_{0} \rangle \cdot  \chi \left(    \langle \bA_{i}, \bG_{0} \rangle \right)  \cdot \mathbb{I}_{ \left\{ \langle \bA_{i}, \bG_{0}   \rangle  \cdot \langle \bA_{i}, \bT_{0} \rangle  \ge 0 \right\} } \Big\vert \nonumber\\
&\le  \frac{1}{m}\sum_{i=1}^{m} \Big\vert \langle \bA_{i}, \bG \rangle \cdot \langle \bA_{i}, \bT \rangle \cdot  \chi \left(    \langle \bA_{i}, \bG \rangle \right)  \cdot \mathbb{I}_{ \left\{ \langle \bA_{i}, \bG   \rangle  \cdot \langle \bA_{i}, \bT \rangle  \ge 0 \right\} }  \nonumber\\
& \quad\quad\quad\quad\quad\quad - \langle \bA_{i}, \bG_{0} \rangle \cdot \langle \bA_{i}, \bT_{0} \rangle \cdot  \chi \left(    \langle \bA_{i}, \bG_{0} \rangle \right)  \cdot \mathbb{I}_{ \left\{ \langle \bA_{i}, \bG_{0}   \rangle  \cdot \langle \bA_{i}, \bT_{0} \rangle  \ge 0 \right\} }\Big\vert \nonumber\\
&\le  \frac{1}{m}\sum_{i=1}^{m}\left\vert \langle \bA_{i}, \bG \rangle \cdot \langle \bA_{i}, \bT \rangle \cdot  \chi \left(    \langle \bA_{i}, \bG \rangle \right) - \langle \bA_{i}, \bG_{0} \rangle \cdot \langle \bA_{i}, \bT_{0} \rangle \cdot  \chi \left(    \langle \bA_{i}, \bG_{0} \rangle \right) \right\vert \nonumber\\
& \le \frac{1}{m}\sum_{i=1}^{m}  \left\vert \langle \bA_{i}, \bG \rangle \cdot  \chi \left(    \langle \bA_{i}, \bG \rangle \right) - \langle \bA_{i}, \bG_{0} \rangle \cdot  \chi \left(    \langle \bA_{i}, \bG_{0} \rangle \right) \right\vert \cdot \left\vert \langle \bA_{i}, \bT \rangle  \right\vert \nonumber\\
& \quad\quad\quad\quad +  \frac{1}{m}\sum_{i=1}^{m} \left\vert   \langle \bA_{i}, \bT - \bT_{0} \rangle \right\vert  \cdot \left\vert  \langle \bA_{i}, \bG_{0} \rangle \cdot  \chi \left(    \langle \bA_{i}, \bG_{0} \rangle \right) \right\vert \nonumber\\
& \le \frac{0.65\alpha_h}{\delta }\left( \frac{1}{m}\sum_{i=1}^{m}  \left\vert \langle \bA_{i}, \bG - \bG_{0} \rangle \right\vert \cdot \left\vert \langle \bA_{i}, \bT \rangle  \right\vert +  \frac{1}{m}\sum_{i=1}^{m} \left\vert   \langle \bA_{i}, \bT - \bT_{0} \rangle \right\vert  \cdot \left\vert  \langle \bA_{i}, \bG_{0} \rangle   \right\vert \right) \label{equ_proof_localcurv_clean_first_lip}\\
& \le \frac{0.65\alpha_h}{\delta }\left(  \frac{1}{\sqrt{m}}\| \mathcal{A}( \bG - \bG_{0} )\|_2  \cdot \frac{1}{\sqrt{m}}\| \mathcal{A}( \bT )\|_2  + \frac{1}{\sqrt{m}}\| \mathcal{A}( \bT - \bT_{0} )\|_2  \cdot \frac{1}{\sqrt{m}}\| \mathcal{A}( \bG_0 )\|_2 \right) \label{equ_proof_localcurv_clean_first_cauchy} \\
& \le \frac{c_2\alpha_h}{\delta } \left( \left\Vert \bG - \bG_{0} \right\Vert_{F} \left\Vert \bT \right\Vert_{F} +  \left\Vert \bT - \bT_{0} \right\Vert_{F}\left\Vert \bG_{0} \right\Vert_{F} \right)\label{equ_proof_localcurv_clean_first_rip}\\
& \le \frac{c_2\alpha_h\tau}{\delta } ,\nonumber
\end{align}
where \eqref{equ_proof_localcurv_clean_first_lip} follows from the Lipschitz property of $t\chi(t)$, \eqref{equ_proof_localcurv_clean_first_cauchy} follows from the Cauchy-Schwarz inequality, and \eqref{equ_proof_localcurv_clean_first_rip} follows from Lemma~\ref{lemma_RIP_A_rec}.

Let $\tau = c_{1}\delta \varepsilon$, then provided  $m \ge c_{2} \varepsilon^{-2} \left(\log{\frac{1}{\delta\varepsilon}}\right) nr $, 
\begin{equation*}
\frac{1}{m}\sum_{i=1}^{m}   \langle \bA_{i}, \bG \rangle \cdot \langle \bA_{i}, \bT \rangle \cdot  \chi \left(    \langle \bA_{i}, \bG \rangle \right)  \cdot \mathbb{I}_{ \left\{ \langle \bA_{i}, \bG   \rangle  \cdot \langle \bA_{i}, \bT \rangle  \ge 0 \right\} }  \ge \gamma_{1}\cdot  \langle \bG,  \bT \rangle  -    1.1\varepsilon \alpha_{h} 
\end{equation*}
holds for all rank-$2r$ matrices $\bG$ and $\bT$ with probability at least $1- \exp{\left(-c\varepsilon^{2}m\right)}$. The proof is finished by setting  $\delta$ arbitrarily small and $\varepsilon=0.001$.

\section{Proof of Proposition~\ref{prop_local_curv_clean_second}}\label{proof_prop_local_curv_clean_second}

First, note that due to the definition of $\mathcal{D}$ in \eqref{def_setD}, $-B_2$ can be written as
\begin{align}
-B_2 & = \frac{1}{2m}\sum_{i \in \mathcal{D}}  \left\vert \langle \bA_{i}, \bU\bV^{T} - \bX\bY^{T}   \rangle  \right\vert \cdot  \left\vert \langle \bA_{i}, \bH_{1}\bV^{T} + \bU\bH_{2}^{T} \rangle \right\vert  \cdot    \mathbb{I}_{ \left\{   \left\vert  \langle \bA_{i}, \bU\bV^{T} - \bX\bY^{T}  \rangle   \right\vert   \le   0.70 \alpha_{h}  \left\Vert \bU\bV^{T}- \bX\bY^{T} \right\Vert_{F}  \right\}  } \nonumber \\
& =   \frac{1}{m}\sum_{i \in \mathcal{D}} \left\vert \langle \bA_{i}, \bU\bV^{T} - \bX\bY^{T}   \rangle\right\vert  \cdot \left\vert\langle \bB_{i}, \bH\bW^{T} \rangle \right\vert \cdot    \mathbb{I}_{ \left\{   \left\vert  \langle \bA_{i}, \bU\bV^{T} - \bX\bY^{T}  \rangle   \right\vert   \le   0.70 \alpha_{h}  \left\Vert \bU\bV^{T}- \bX\bY^{T} \right\Vert_{F}  \right\}  } \nonumber\\
& \leq  0.70 \alpha_{h} \left\Vert \bU\bV^{T}- \bX\bY^{T} \right\Vert_{F} \cdot \frac{1}{m}\sum_{i \in \mathcal{D}} \left\vert\langle \bB_{i}, \bH\bW^{T} \rangle \right\vert .\label{equ_proof_b1}
\end{align}
Note that when $i\in\mathcal{D}$, we have the following lemma, whose proof is given in Appendix~\ref{proof_lemma_setD}.
\begin{lemma}\label{lemma_setD}
If $i\in\mathcal{D}$, then $\left\vert \langle \bB_i, \bH\bW^T \rangle \right\vert  <   \frac{1}{2}   \left\vert \langle \bB_i,  \bH\bH^T \rangle  \right\vert$. 
\end{lemma}
Plugging Lemma~\ref{lemma_setD} into \eqref{equ_proof_b1}, we obtain
\begin{align}
-B_2 & \le  0.35 \alpha_{h} \left\Vert \bU\bV^{T}- \bX\bY^{T} \right\Vert_{F} \cdot \frac{1}{m}\sum_{i \in \mathcal{D}} \left\vert\langle \bB_{i}, \bH\bH^{T} \rangle \right\vert \nonumber \\
& \le 0.35 \alpha_{h}  \left\Vert \bU\bV^{T}- \bX\bY^{T} \right\Vert_{F} \frac{1}{m} \sqrt{m} \left(\sum_{i\in\mathcal{D}} \left\vert\langle \bA_{i}, \bH_{1}\bH_{2}^{T} \rangle \right\vert^{2} \right)^{1/2} \label{equ_B2_local_curv_cauchy}\\
&\leq 0.35 \alpha_{h}  \left\Vert \bU\bV^{T}- \bX\bY^{T} \right\Vert_{F} \frac{1}{\sqrt{m}}\|\mathcal{A}(\bH_{1}\bH_{2}^T)\|_2 \nonumber \\ 
& \le   0.35\left(1+\delta\right) \alpha_{h}  \left\Vert \bU\bV^{T}- \bX\bY^{T} \right\Vert_{F}  \left\Vert \bH_{1}\bH_{2}^{T} \right\Vert_{F},
\end{align}
where \eqref{equ_B2_local_curv_cauchy} follows from the Cauchy-Schwarz inequality and the last inequality follows from Lemma~\ref{lemma_RIP_A_rec}.

\section{Proof of Proposition~\ref{prop_outlier_curvature}}\label{proof_outlier_curvature}
First, note that by the definitions of $\mathcal{E}_{i}$ and $\tilde{\mathcal{E}}_i$, we have 
\begin{align*}
\left| \left(\mathcal{A}_{i}\left(\bU\bV^{T} \right) - y_{i}\right)\mathbb{I}_{ \mathcal{E}_{i}} \right| & \leq \alpha_h \mathrm{med}\left(  \left\vert \by - \mathcal{A}\left( \bU\bV^{T} \right)     \right\vert  \right), \\
\left| \left(\mathcal{A}_{i}\left(\bU\bV^{T} \right) - \mathcal{A}_{i}\left(\bX\bY^{T} \right) \right)\mathbb{I}_{ \tilde{\mathcal{E}}_{i}} \right| & \leq  \alpha_h \mathrm{med}\left(  \left\vert \by - \mathcal{A}\left( \bU\bV^{T} \right)     \right\vert  \right).
\end{align*}
Then we further obtain
\begin{align}
 |\langle \nabla^{o} f_{tr}\left(\bW\right),  \bH  \rangle|
& \le \frac{1}{m} \sum_{i \in \mathcal{S}}  \left\vert  \left[ \left(\mathcal{B}_{i}\left(\bW\bW^{T} \right) - y_{i}\right)\mathbb{I}_{ \mathcal{E}_{i}}  -  \left(\mathcal{B}_{i}\left(\bW\bW^{T} \right) - \mathcal{B}_{i}\left(\bZ\bZ^{T} \right) \right)\mathbb{I}_{ \tilde{\mathcal{E}}_{i}}  \right] \langle \bB_{i}, \bH\bW^{T} \rangle  \right\vert \nonumber\\
& =  \frac{1}{m}   \sum_{i \in \mathcal{S}}  \left\vert \left[ \left(\mathcal{A}_{i}\left(\bU\bV^{T} \right) - y_{i}\right)\mathbb{I}_{ \mathcal{E}_{i}}  -  \left(\mathcal{A}_{i}\left(\bU\bV^{T} \right) - \mathcal{A}_{i}\left(\bX\bY^{T} \right) \right)\mathbb{I}_{ \tilde{\mathcal{E}}_{i}}  \right]  \langle  \bB_{i}, \bH\bW^{T} \rangle \right\vert \nonumber\\
& \leq \frac{2\alpha_{h}}{m}\mathrm{med}\left(  \left\vert \by - \mathcal{A}\left( \bU\bV^{T} \right)     \right\vert  \right) \sum_{i\in\cS}   \left|  \langle  \bB_{i}, \bH\bW^{T} \rangle\right| \nonumber \\
& \leq \frac{2\alpha_{h}}{m}\mathrm{med}\left(  \left\vert \by - \mathcal{A}\left( \bU\bV^{T} \right)     \right\vert  \right) \sqrt{|\cS|} \left(\sum_{i\in\cS} \left|  \langle  \bB_{i}, \bH\bW^{T} \rangle\right|^2 \right)^{1/2} \label{equ_local_curv_outlier_cauchy} \\
& \leq \alpha_{h} \sqrt{ \frac{|\cS|}{m}} \mathrm{med}\left(  \left\vert \by - \mathcal{A}\left( \bU\bV^{T} \right)     \right\vert  \right)  \left(\frac{1}{m}\sum_{i=1}^m \left|  \langle  \bA_{i}, \bH_{1}\bV^{T}+\bU\bH_{2}^{T} \rangle\right|^2 \right)^{1/2} \nonumber \\
&\leq  0.70\alpha_{h} \sqrt{s}\left\Vert \bX\bY^{T} - \bU\bV^{T} \right\Vert_{F} \cdot  (1+\delta) \| \bH_{1}\bV^T + \bU\bH_{2}^{T}\|_F \label{equ_local_curv_outlier_rip} \\
& \leq 0.71\alpha_{h}\sqrt{s} \left\Vert \bX\bY^{T} - \bU\bV^{T} \right\Vert_{F} \| \bH_{1}\bV^T + \bU\bH_{2}^{T}\|_F, \nonumber
\end{align}
where \eqref{equ_local_curv_outlier_cauchy} follows from the Cauchy-Schwarz inequality, \eqref{equ_local_curv_outlier_rip} follows from \eqref{equ_sample_median_concen} and Lemma~\ref{lemma_RIP_A_rec}, and the last inequality follows by setting $\delta$ sufficiently small.

\section{Proof of Proposition~\ref{prop_local_smooth_clean}}\label{proof_prop_local_smooth_clean}

Since $ \left\Vert \nabla f_{tr}\left(\bW\right)  \right\Vert_{F}^{2} = \max_{\left\Vert \bG \right\Vert_{F} = 1} \left\vert  \langle  \nabla f_{tr}\left(\bW\right) , \bG \rangle \right\vert^{2}$, it is sufficient to upper bound $\left\vert  \langle  \nabla f_{tr}\left(\bW\right) , \bG \rangle \right\vert^{2}$ for any arbitrary $\bG = \begin{bmatrix} \bG_{1}^{T} & \bG_{2}^{T}\end{bmatrix}^{T}  \in\mathbb{R}^{(n_{1}+n_{2}) \times r}$ with $\bG_{1}\in\mathbb{R}^{n_{1}\times r}$ and $\bG_{2}\in\mathbb{R}^{n_{2}\times r}$ satisfying $\left\Vert \bG \right\Vert_{F} = 1$. We have
\begin{align}
&\left\vert  \langle  \nabla f_{tr}\left(\bW\right) , \bG \rangle \right\vert^{2}\nonumber\\
& = \left\vert  \langle \frac{1}{m} \sum_{i=1}^{m} \left( \mathcal{B}_{i}\left(\bW\bW^{T} \right) - y_{i}\right)\bB_{i}\bW\mathbb{I}_{ \mathcal{E}_{i}} , \bG\rangle \right\vert^{2}  \nonumber\\
& = \left\vert  \left \langle \frac{1}{m} \sum_{i=1}^{m} \left( \mathcal{A}_{i}\left(\bU\bV^{T}\right) - y_{i}\right)\bB_{i}\bW\mathbb{I}_{ \mathcal{E}_{i}}, \bG \right\rangle \right\vert^{2} \nonumber\\
& = \left\vert   \frac{1}{m} \sum_{i=1}^{m} \left(\langle \bA_i, \bU\bV^{T}\rangle  - y_{i} \right)  \cdot \langle \bB_{i},\bG\bW^{T} \rangle \cdot \mathbb{I}_{ \mathcal{E}_{i}} \right\vert^{2} \nonumber \\
& \le \left( \frac{1}{m} \sum_{i=1}^{m} \left(\langle \bA_i, \bU\bV^{T}\rangle - y_{i} \right)^{2}  \cdot \mathbb{I}_{ \mathcal{E}_{i}} \right)  \cdot \left( \frac{1}{m} \sum_{i=1}^{m}  \left| \langle \bA_{i}, \frac{1}{2}\left(\bG_{1}\bV^{T} + \bU\bG_{2}^{T}\right) \rangle\right|^{2}  \right) \label{equ_local_smooth_cau_sch}
\end{align}
where \eqref{equ_local_smooth_cau_sch} follows from the Cauchy-Schwarz inequality. Due to \eqref{truncation_bound}, we have
\begin{align}
 \frac{1}{m} \sum_{i=1}^{m} \left(\langle \bA_i, \bU\bV^{T}\rangle  - y_{i} \right)^{2}  \cdot \mathbb{I}_{ \mathcal{E}_{i}}  & \leq  \frac{1}{m} \sum_{i=1}^{m} \left(\langle \bA_i, \bU\bV^{T}\rangle  - y_{i} \right)^{2}  \cdot \mathbb{I}_{ \left\{ \left\vert \langle \bA_i, \bU\bV^{T}\rangle  - y_i  \right\vert \le 0.70 \alpha_{h}  \left\Vert \bU\bV^{T}- \bX\bY^{T} \right\Vert_{F}   \right\} } \nonumber \\
 & \leq 0.70^2\alpha_h^2\left\Vert \bU\bV^{T}- \bX\bY^{T} \right\Vert_{F} ^2.\label{first_term_prop4}
\end{align}
From Lemma~\ref{lemma_RIP_A_rec}, we have
\begin{align}
 \frac{1}{m} \sum_{i=1}^{m}  \left| \langle \bA_{i},\frac{1}{2}\left(\bG_{1}\bV^{T} + \bU\bG_{2}^{T}\right) \rangle\right|^{2} 
 & \leq \frac{1}{4}(1+\delta)^2 \left\Vert \bG_{1}\bV^{T} + \bU\bG_{2}^{T}\right\Vert_{F}^{2} \nonumber\\
 & \le  \frac{1}{2}(1+\delta)^2 \left( \left\Vert \bG_{1}\bV^{T} \right\Vert_{F}^{2} +  \left\Vert  \bU\bG_{2}^{T}\right\Vert_{F}^{2}\right)\nonumber\\
 & \le \frac{1}{2}(1+\delta)^2  \max\left\{ \left\Vert \bU\right\Vert^{2}, \left\Vert \bV \right\Vert^{2} \right\}\nonumber\\
 & \le \frac{1}{2}(1+\delta)^2 \left\Vert \bW \right\Vert^{2}. \label{second_term_prop4}
\end{align}
Plugging \eqref{first_term_prop4} and \eqref{second_term_prop4} into \eqref{equ_local_smooth_cau_sch}, we have
\begin{align*}
\left\vert  \langle  \nabla f_{tr}\left(\bW\right) , \bG \rangle \right\vert^{2}
&  \le  \frac{1}{2} \cdot 0.70^2 \left(1+\delta\right)^{2} \alpha_{h}^{2}  \left\Vert \bU\bV^{T} - \bX\bY^{T} \right\Vert_{F}^{2} \left\Vert \bW \right\Vert^{2},
\end{align*}
and the proof is completed by setting $\delta$ small enough.


\section{Proof of Proposition~\ref{bound_Y_noisy}}\label{proof_bound_Y_noisy}

First, consider the bound of $\left\Vert \bK_{1}  - \mathbb{E}\left[\bK_{1}\right] \right\Vert $. Define
\begin{align*}
\bS_{i} = \mathcal{A}_{i}\left( \bM \right) \bA_i \mathbb{I}_{\left\{ \left\vert \mathcal{A}_{i}\left( \bM \right) \right\vert \le   \alpha_{y} C_{M} \right\}}    -     \gamma_{2}   \bM,\quad i\in\mathcal{S}_1^c, 
\end{align*}
which satisfies $\mathbb{E}\left[ \bS_{i} \right] = \boldsymbol{0}$, and $\bK_{1} -  \mathbb{E}\left[ \bK_{1} \right] =  \frac{1}{\left\vert \mathcal{S}_{1}^{c} \right\vert}  \sum_{i\in \mathcal{S}_{1}^{c}} \bS_{i}$.

Based on \cite[Proposition 5.34]{Vershynin2012}, we know 
\begin{equation*}
\mathbb{P}\left\{ \left\vert \left\Vert \bA_{i} \right\Vert - \mathbb{E}\left[ \left\Vert \bA_{i} \right\Vert  \right] \right\vert > t \right\} \le 2 e^{-t^{2}/2},
\end{equation*}
which shows $\left\Vert \bA_{i} \right\Vert - \mathbb{E}  \left\Vert \bA_{i} \right\Vert  $ is a  sub-Gaussian random variable satisfying $\left\Vert \left\Vert \bA_{i} \right\Vert - \mathbb{E}  \left\Vert \bA_{i} \right\Vert  \right\Vert_{\psi_{2}} \le c$. Then, we have $\left\Vert \bA_{i} \right\Vert_{\psi_{2}} \leq  \mathbb{E}  \left\Vert \bA_{i} \right\Vert  + c \le 2\sqrt{n} + c$, where the last inequality follows from the fact $\mathbb{E}  \left\Vert \bA_{i} \right\Vert  \le 2\sqrt{n}$. As a result, we can calculate 
\begin{align*}
\left\Vert \bS_{i} \right\Vert_{\psi_{2}}
& = \left\Vert   \mathcal{A}_{i}\left( \bM \right)  \bA_i \mathbb{I}_{\left\{ \left\vert \mathcal{A}_{i}\left( \bM \right) \right\vert \le   \alpha_{y} C_{M} \right\}}    -     \gamma_{2}   \bM \right\Vert_{\psi_{2}}\\
& \le \left\Vert   \mathcal{A}_{i}\left( \bM \right)  \bA_i \mathbb{I}_{\left\{ \left\vert \mathcal{A}_{i}\left( \bM \right)  \right\vert \le   \alpha_{y} C_{M} \right\}}  \right\Vert_{\psi_{2}} +   \gamma_{2}    \left\Vert     \bM \right\Vert \\
& \le \alpha_{y} C_{M} \left\Vert \bA_{i} \right\Vert_{\psi_{2}} +  \gamma_{2}    \left\Vert \bM \right\Vert \le c_{1} \sqrt{n} \alpha_{y} \left\Vert \bM \right\Vert_{F},
\end{align*}
where $c_1$ is some constant. Moreover, we have  
\begin{align*}
\sigma_{\bS_{i}}^{2}  
&:= \max\left\{  \left\Vert \frac{1}{\left\vert \mathcal{S}_{1}^{c} \right\vert} \sum_{i\in\mathcal{S}_{1}^{c}}  \mathbb{E}\left[ \bS_{i}\bS_{i}^{T} \right]  \right\Vert , \left\Vert \frac{1}{\left\vert \mathcal{S}_{1}^{c} \right\vert} \sum_{i\in\mathcal{S}_{1}^{c}}  \mathbb{E}\left[ \bS_{i}^{T}\bS_{i} \right]  \right\Vert \right\}\\
& = \max\Big\{ \left\Vert  \mathbb{E}\left[ \left(   \mathcal{A}_{i}\left( \bM \right) \bA_i \mathbb{I}_{\left\{ \left\vert \mathcal{A}_{i}\left( \bM \right) \right\vert \le   \alpha_{y} C_{M} \right\}}    -      \gamma_{2}    \bM\right)\left(   \mathcal{A}_{i}\left( \bM \right) \bA_i \mathbb{I}_{\left\{ \left\vert \mathcal{A}_{i}\left( \bM \right) \right\vert \le   \alpha_{y} C_{M} \right\}}    -      \gamma_{2}    \bM\right)^{T} \right] \right\Vert ,\\
&\quad\quad\quad\quad  \left\Vert  \mathbb{E}\left[ \left(   \mathcal{A}_{i}\left( \bM \right) \bA_i \mathbb{I}_{\left\{ \left\vert \mathcal{A}_{i}\left( \bM \right) \right\vert \le   \alpha_{y} C_{M} \right\}}    -      \gamma_{2}    \bM\right)^{T}  \left(   \mathcal{A}_{i}\left( \bM \right) \bA_i \mathbb{I}_{\left\{ \left\vert \mathcal{A}_{i}\left( \bM \right) \right\vert \le   \alpha_{y} C_{M} \right\}}    -      \gamma_{2}    \bM\right) \right] \right\Vert \Big\}\\
& = \max\Big\{ \left\Vert  \mathbb{E}\left[  \left(\mathcal{A}_{i}\left( \bM \right)  \right)^{2} \bA_i\bA_{i}^{T} \mathbb{I}_{\left\{ \left\vert \mathcal{A}_{i}\left( \bM \right)  \right\vert \le   \alpha_{y} C_{M} \right\}} \right]  -  \gamma_{2}^{2}\bM\bM^{T} \right\Vert, \\
&\quad\quad\quad\quad \left\Vert  \mathbb{E}\left[  \left(\mathcal{A}_{i}\left( \bM \right)  \right)^{2} \bA_{i}^{T}\bA_i \mathbb{I}_{\left\{ \left\vert \mathcal{A}_{i}\left( \bM \right)  \right\vert \le   \alpha_{y} C_{M} \right\}} \right]  -  \gamma_{2}^{2}\bM^{T}\bM \right\Vert  \Big\}\\
& \le \max\Big\{ \left\Vert \mathbb{E}\left[ \left(\mathcal{A}_{i}\left( \bM \right)  \right)^{2} \bA_i\bA_{i}^{T} \mathbb{I}_{\left\{ \left\vert \mathcal{A}_{i}\left( \bM \right)  \right\vert \le   \alpha_{y} C_{M} \right\}} \right] \right\Vert +  \gamma_{2}^2 \left\Vert \bM \right\Vert^{2}, \\
&\quad\quad\quad\quad \left\Vert \mathbb{E}\left[ \left(\mathcal{A}_{i}\left( \bM \right)  \right)^{2} \bA_{i}^{T}\bA_i \mathbb{I}_{\left\{ \left\vert \mathcal{A}_{i}\left( \bM \right)  \right\vert \le   \alpha_{y} C_{M} \right\}} \right] \right\Vert +  \gamma_{2}^2 \left\Vert \bM \right\Vert^{2} \Big\}\\
& \le \alpha_{y}^{2} C_{M}^{2} \max\left\{ \left\Vert \mathbb{E}\left[ \bA_{i}\bA_{i}^{T} \right] \right\Vert,  \left\Vert \mathbb{E}\left[ \bA_{i}^{T}\bA_{i} \right] \right\Vert \right\} + \gamma_{2}^{2} \left\Vert \bM \right\Vert^{2}\\
& \le c_{2} n \alpha_{y}^{2}\left\Vert\bM\right\Vert_{F}^{2},
\end{align*}
where $c_2$ is some constant. By Lemma~\ref{lemma_orlicz_matrix_bernstein}, we have
\begin{align*}
\left\Vert  \frac{1}{\left\vert \mathcal{S}_{1}^{c} \right\vert} \sum_{i\in\mathcal{S}_{1}^{c}} \bS_{i} \right\Vert
& \le C \sqrt{n} \alpha_{y} \left\Vert \bM \right\Vert_{F} \max\Bigg\{  \sqrt{\frac{t + \log{\left(2n\right)}}{\left\vert \mathcal{S}_{1}^{c} \right\vert}}, \frac{t + \log{\left(2n\right)}}{\left\vert \mathcal{S}_{1}^{c} \right\vert} \Bigg\},
\end{align*}
with probability at least $1-e^{-t}$, where $C$ is some constant. Set $t=c\log n$. As long as $|\mathcal{S}_1^c| = (1-s_1) m/2 \geq c'\log (n)$, we have
\begin{equation}\label{equ_bound_Y1_spectral}
\left\Vert \bK_{1}  - \mathbb{E}\left[\bK_{1}\right] \right\Vert \leq  C \alpha_y \|\bM\|_F  \sqrt{\frac{n\log n}{m}}
\end{equation}
holds with probability at least $1-n^{-c}$ for some $c>1$.

Next, we employ the same technique to bound $\left\Vert  \bK_{2}  - \mathbb{E}\left[\bK_{2}\right] \right\Vert $. Define $\bT_{i} = y_i \bA_i \mathbb{I}_{\left\{ \left\vert y_i \right\vert  \le  \alpha_{y} C_{M} \right\}} $, which satisfies $\mathbb{E}[\bT_i] = \boldsymbol{0}$ and $\bK_{2} - \mathbb{E}\left[\bK_{2}\right] = \frac{1}{\left\vert \mathcal{S}_{1} \right\vert} \sum_{i\in \mathcal{S}_{1}} \bT_{i}$. We have 
\begin{align*}
\left\Vert \bT_{i} \right\Vert_{\psi_{2}}
& = \left\Vert y_i \bA_i \mathbb{I}_{\left\{ \left\vert y_i \right\vert  \le  \alpha_{y} C_{M} \right\}} \right\Vert_{\psi_{2}} \le \alpha_{y}C_{M} \left\Vert \bA_{i} \right\Vert_{\psi_{2}}\le c_{1}\sqrt{n}\alpha_{y}\left\Vert\bM\right\Vert_{F},
\end{align*}
where $c_1$ is some constant, and 
\begin{align*}
\sigma_{\bT_{i}}^{2} 
&: =  \max\left\{ \left\Vert \frac{1}{\left\vert \mathcal{S}_{1} \right\vert} \sum_{i\in\mathcal{S}_{1}}  \mathbb{E}\left[ \bT_{i}\bT_{i}^{T} \right]  \right\Vert,  \left\Vert \frac{1}{\left\vert \mathcal{S}_{1} \right\vert} \sum_{i\in\mathcal{S}_{1}}  \mathbb{E}\left[ \bT_{i}^{T}\bT_{i} \right]  \right\Vert \right\}\\
& \le \alpha_{y}^{2} C_{M}^{2} \max\left\{ \left\Vert \mathbb{E}\left[ \bA_i\bA_{i}^{T} \right] \right\Vert, \left\Vert \mathbb{E}\left[ \bA_i^{T}\bA_{i} \right] \right\Vert \right\} \le c_{2}n\alpha_{y}^{2}\left\Vert\bM\right\Vert_{F}^{2},
\end{align*}
where $c_2$ is some constant. Again, by Lemma~\ref{lemma_orlicz_matrix_bernstein} we have
\begin{align*}
\left\Vert \frac{1}{\left\vert \mathcal{S}_{1} \right\vert} \sum_{i\in\mathcal{S}_{1}} \bT_{i} \right\Vert 
& \le C \sqrt{n} \alpha_{y} \left\Vert\bM\right\Vert_{F} \max\Bigg\{ \sqrt{\frac{t + \log{\left(2n\right)}}{\left\vert\mathcal{S}_{1}\right\vert}},\frac{t + \log{\left(2n\right)}}{\left\vert\mathcal{S}_{1}\right\vert} \Bigg\}
\end{align*}
with probability at least $1-e^{-t}$. Then by setting $t=c\log n$, and recalling $|\mathcal{S}_1 | = s_1 m/2 $, we have with probability at least $1-n^{-c}$,
\begin{equation}\label{equ_bound_Y2_spectral}
\left\| \bK_2 \right\| \le C\sqrt{n}  \alpha_{y} \left\Vert\bM\right\Vert_{F} \max\Bigg\{ \sqrt{\frac{\log n}{s_1 m}}, \frac{\log n}{s_1 m} \Bigg\}.
\end{equation}

Combing \eqref{equ_bound_Y1_spectral} and \eqref{equ_bound_Y2_spectral}, we have with probability at least $1-n^{-c}$,
\begin{align*}
&\left\Vert \bK - \left(1-s_1\right)  \gamma_{2} \bM \right\Vert  \\
& \le \left(1-s_1\right) \left\Vert \bK_{1} -  \gamma_{2}  \bM \right\Vert + s_1 \left\Vert\bK_{2}\right\Vert \\
& \leq  C \alpha_y \|\bM\|_F  \sqrt{\frac{n\log n}{m}} + C\sqrt{n}  \alpha_{y} \left\Vert\bM\right\Vert_{F} \max\Bigg\{ \sqrt{\frac{s_1\log n}{ m}}, \frac{\log n}{ m} \Bigg\} \\
& \leq  C \alpha_y \|\bM\|_F  \sqrt{\frac{n\log n}{m}} 
\end{align*}
provided that $m>c_2 \log n$ for large enough $c_2$.



\section{Proof of Lemma~\ref{lemma_setD}}\label{proof_lemma_setD}

Since $\bH= \bW-\bZ\bQ$, we can write
\begin{align*}
\langle \bB_{i}, \bW\bW^{T} - \bZ\bZ^{T}   \rangle & = \langle \bB_{i}, \bW\bW^{T} - (\bZ\bQ)(\bZ\bQ)^{T}   \rangle \\
& =  \langle \bB_{i}, \bW\bW^{T} - (\bW-\bH)(\bW-\bH)^{T}   \rangle \\
& =2 \langle \bB_i, \bH\bW^T \rangle  -\langle \bB_i,\bH\bH^T\rangle .
\end{align*}
Therefore,  $i\in\mathcal{D} $ if and only if
\begin{equation} \label{equivalent_D}
\left(2 \langle \bB_i, \bH\bW^T \rangle  -\langle \bB_i,\bH\bH^T\rangle  \right)  \langle \bB_{i}, \bH\bW^{T} \rangle<0 .
\end{equation}
If $\langle \bB_{i}, \bH\bW^{T} \rangle>0$, then $ \langle \bB_i, \bH\bW^T \rangle  <\frac{1}{2} \langle \bB_i,\bH\bH^T\rangle$; if $\langle \bB_{i}, \bH\bW^{T} \rangle<0$, then $ \langle \bB_i, \bH\bW^T \rangle  >\frac{1}{2} \langle \bB_i,\bH\bH^T\rangle$. Therefore, we have $|\langle \bB_i, \bH\bW^T \rangle|  <\frac{1}{2} |\langle \bB_i,\bH\bH^T\rangle|$.


\bibliographystyle{IEEEtran} 
\bibliography{Outlier_medianTWF}

\begin{thebibliography}{10}
\providecommand{\url}[1]{#1}
\csname url@samestyle\endcsname
\providecommand{\newblock}{\relax}
\providecommand{\bibinfo}[2]{#2}
\providecommand{\BIBentrySTDinterwordspacing}{\spaceskip=0pt\relax}
\providecommand{\BIBentryALTinterwordstretchfactor}{4}
\providecommand{\BIBentryALTinterwordspacing}{\spaceskip=\fontdimen2\font plus
\BIBentryALTinterwordstretchfactor\fontdimen3\font minus
  \fontdimen4\font\relax}
\providecommand{\BIBforeignlanguage}[2]{{%
\expandafter\ifx\csname l@#1\endcsname\relax
\typeout{** WARNING: IEEEtran.bst: No hyphenation pattern has been}%
\typeout{** loaded for the language `#1'. Using the pattern for}%
\typeout{** the default language instead.}%
\else
\language=\csname l@#1\endcsname
\fi
#2}}
\providecommand{\BIBdecl}{\relax}
\BIBdecl

\bibitem{recht2010guaranteed}
B.~Recht, M.~Fazel, and P.~A. Parrilo, ``Guaranteed minimum-rank solutions of
  linear matrix equations via nuclear norm minimization,'' \emph{SIAM Review},
  vol.~52, no.~3, pp. 471--501, 2010.

\bibitem{gross2011recovering}
D.~Gross, ``Recovering low-rank matrices from few coefficients in any basis,''
  \emph{IEEE Transactions on Information Theory}, vol.~57, no.~3, pp.
  1548--1566, March 2011.

\bibitem{negahban2011estimation}
S.~Negahban and M.~J. Wainwright, ``Estimation of (near) low-rank matrices with
  noise and high-dimensional scaling,'' \emph{The Annals of Statistics},
  vol.~39, no.~2, pp. 1069--1097, 2011.

\bibitem{candes2012exact}
E.~Candes and B.~Recht, ``Exact matrix completion via convex optimization,''
  \emph{Communications of the ACM}, vol.~55, no.~6, pp. 111--119, 2012.

\bibitem{chen2014robust}
Y.~Chen and Y.~Chi, ``Robust spectral compressed sensing via structured matrix
  completion,'' \emph{IEEE Transactions on Information Theory}, vol.~60,
  no.~10, pp. 6576--6601, 2014.

\bibitem{chen2015exact}
Y.~Chen, Y.~Chi, and A.~Goldsmith, ``Exact and stable covariance estimation
  from quadratic sampling via convex programming,'' \emph{IEEE Transactions on
  Information Theory}, vol.~61, no.~7, pp. 4034--4059, July 2015.

\bibitem{davenport2016overview}
M.~A. Davenport and J.~Romberg, ``An overview of low-rank matrix recovery from
  incomplete observations,'' \emph{IEEE Journal of Selected Topics in Signal
  Processing}, vol.~10, no.~4, pp. 608--622, 2016.

\bibitem{jain2010guaranteed}
P.~Jain, R.~Meka, and I.~S. Dhillon, ``Guaranteed rank minimization via
  singular value projection,'' in \emph{Advances in Neural Information
  Processing Systems (NIPS)}, 2010, pp. 937--945.

\bibitem{candes2011tight}
E.~J. Cand\`es and Y.~Plan, ``Tight oracle inequalities for low-rank matrix
  recovery from a minimal number of noisy random measurements,'' \emph{IEEE
  Transactions on Information Theory}, vol.~57, no.~4, pp. 2342--2359, 2011.

\bibitem{burer2003nonlinear}
S.~Burer and R.~D. Monteiro, ``A nonlinear programming algorithm for solving
  semidefinite programs via low-rank factorization,'' \emph{Mathematical
  Programming}, vol.~95, no.~2, pp. 329--357, 2003.

\bibitem{zheng2015convergent}
Q.~Zheng and J.~Lafferty, ``A convergent gradient descent algorithm for rank
  minimization and semidefinite programming from random linear measurements,''
  in \emph{Advances in Neural Information Processing Systems (NIPS)}, 2015.

\bibitem{tu2016low}
S.~Tu, R.~Boczar, M.~Simchowitz, M.~Soltanolkotabi, and B.~Recht, ``Low-rank
  solutions of linear matrix equations via procrustes flow,'' in
  \emph{Proceedings of the 33rd International Conference on International
  Conference on Machine Learning (ICML)}, 2016, pp. 964--973.

\bibitem{zhao2015nonconvex}
T.~Zhao, Z.~Wang, and H.~Liu, ``A nonconvex optimization framework for low rank
  matrix estimation,'' in \emph{Advances in Neural Information Processing
  Systems (NIPS)}, 2015, pp. 559--567.

\bibitem{chen2015fast}
Y.~Chen and M.~J. Wainwright, ``Fast low-rank estimation by projected gradient
  descent: General statistical and algorithmic guarantees,'' \emph{arXiv
  preprint arXiv:1509.03025}, 2015.

\bibitem{park2016provable}
D.~Park, A.~Kyrillidis, S.~Bhojanapalli, C.~Caramanis, and S.~Sanghavi,
  ``Provable burer-monteiro factorization for a class of norm-constrained
  matrix problems,'' \emph{arXiv preprint arXiv:1606.01316}, 2016.

\bibitem{jain2013low}
P.~Jain, P.~Netrapalli, and S.~Sanghavi, ``Low-rank matrix completion using
  alternating minimization,'' in \emph{Proceedings of the Forty-fifth Annual
  ACM Symposium on Theory of Computing}, 2013, pp. 665--674.

\bibitem{hardt2014understanding}
M.~Hardt, ``Understanding alternating minimization for matrix completion,'' in
  \emph{IEEE 55th Annual Symposium on Foundations of Computer Science (FOCS)},
  2014, pp. 651--660.

\bibitem{bhojanapalli2016global}
S.~Bhojanapalli, B.~Neyshabur, and N.~Srebro, ``Global optimality of local
  search for low rank matrix recovery,'' \emph{arXiv preprint
  arXiv:1605.07221}, 2016.

\bibitem{ge2016matrix}
R.~Ge, J.~D. Lee, and T.~Ma, ``Matrix completion has no spurious local
  minimum,'' in \emph{Advances in Neural Information Processing Systems
  (NIPS)}, 2016, pp. 2973--2981.

\bibitem{li2016nonconvex}
Q.~Li and G.~Tang, ``The nonconvex geometry of low-rank matrix optimizations
  with general objective functions,'' \emph{arXiv preprint arXiv:1611.03060},
  2016.

\bibitem{li2016symmetry}
X.~Li, Z.~Wang, J.~Lu, R.~Arora, J.~Haupt, H.~Liu, and T.~Zhao, ``Symmetry,
  saddle points, and global geometry of nonconvex matrix factorization,''
  \emph{arXiv preprint arXiv:1612.09296}, 2016.

\bibitem{ge2015escaping}
R.~Ge, F.~Huang, C.~Jin, and Y.~Yuan, ``Escaping from saddle points---online
  stochastic gradient for tensor decomposition,'' \emph{arXiv preprint
  arXiv:1503.02101}, 2015.

\bibitem{lee2016gradient}
J.~D. Lee, M.~Simchowitz, M.~I. Jordan, and B.~Recht, ``Gradient descent only
  converges to minimizers,'' in \emph{Conference on Learning Theory}, 2016, pp.
  1246--1257.

\bibitem{jin2017escape}
C.~Jin, R.~Ge, P.~Netrapalli, S.~M. Kakade, and M.~I. Jordan, ``How to escape
  saddle points efficiently,'' \emph{arXiv preprint arXiv:1703.00887}, 2017.

\bibitem{li2017lowrank}
Y.~Li, Y.~Sun, and Y.~Chi, ``Low-rank positive semidefinite matrix recovery
  from corrupted rank-one measurements,'' \emph{IEEE Transactions on Signal
  Processing}, vol.~65, no.~2, pp. 397--408, Jan 2017.

\bibitem{wright2013compressive}
J.~Wright, A.~Ganesh, K.~Min, and Y.~Ma, ``Compressive principal component
  pursuit,'' \emph{Information and Inference}, vol.~2, no.~1, pp. 32--68, 2013.

\bibitem{huber2011robust}
P.~J. Huber, \emph{Robust Statistics}.\hskip 1em plus 0.5em minus 0.4em\relax
  Springer, 2011.

\bibitem{zhang2016provable}
H.~Zhang, Y.~Chi, and Y.~Liang, ``Provable non-convex phase retrieval with
  outliers: Median truncated {W}irtinger flow,'' \emph{arXiv preprint
  arXiv:1603.03805}, 2016.

\bibitem{candes2015phase}
E.~J. Cand\`es, X.~Li, and M.~Soltanolkotabi, ``Phase retrieval via wirtinger
  flow: Theory and algorithms,'' \emph{IEEE Transactions on Information
  Theory}, vol.~61, no.~4, pp. 1985--2007, 2015.

\bibitem{zhang2016reshaped}
H.~{Zhang}, Y.~{Zhou}, Y.~{Liang}, and Y.~{Chi}, ``Reshaped {W}irtinger flow
  and incremental algorithm for solving quadratic system of equations,''
  \emph{ArXiv 1605.07719}, May 2016.

\bibitem{chen2015solving}
Y.~Chen and E.~Candes, ``Solving random quadratic systems of equations is
  nearly as easy as solving linear systems,'' in \emph{Advances in Neural
  Information Processing Systems (NIPS)}, 2015.

\bibitem{li2016rapid}
X.~Li, S.~Ling, T.~Strohmer, and K.~Wei, ``Rapid, robust, and reliable blind
  deconvolution via nonconvex optimization,'' \emph{arXiv preprint
  arXiv:1606.04933}, 2016.

\bibitem{tibshirani2008fast}
R.~J. Tibshirani, ``Fast computation of the median by successive binning,''
  \emph{arXiv preprint arXiv:0806.3301}, 2008.

\bibitem{yi2016fast}
X.~Yi, D.~Park, Y.~Chen, and C.~Caramanis, ``Fast algorithms for robust pca via
  gradient descent,'' in \emph{Advances in neural information processing
  systems}, 2016, pp. 4152--4160.

\bibitem{cherapanamjeri2016nearly}
Y.~Cherapanamjeri, K.~Gupta, and P.~Jain, ``Nearly-optimal robust matrix
  completion,'' \emph{arXiv preprint arXiv:1606.07315}, 2016.

\bibitem{zhang2017nonconvex}
X.~Zhang, L.~Wang, and Q.~Gu, ``A nonconvex free lunch for low-rank plus sparse
  matrix recovery,'' \emph{arXiv preprint arXiv:1702.06525}, 2017.

\bibitem{sanghavi2016local}
S.~Sanghavi, R.~Ward, and C.~D. White, ``The local convexity of solving systems
  of quadratic equations,'' \emph{Results in Mathematics}, pp. 1--40, 2016.

\bibitem{li2017nonconvex}
Y.~Li, Y.~Chi, H.~Zhang, and Y.~Liang, ``Non-convex low-rank matrix recovery
  from corrupted random linear measurements,'' in \emph{2017 International
  Conference on Sampling Theory and Applications (SampTA)}, 2017.

\bibitem{zheng2016convergence}
Q.~Zheng and J.~Lafferty, ``Convergence analysis for rectangular matrix
  completion using burer-monteiro factorization and gradient descent,''
  \emph{arXiv preprint arXiv:1605.07051}, 2016.

\bibitem{koltchinskii2011nuclear}
V.~Koltchinskii, K.~Lounici, and A.~B. Tsybakov, ``Nuclear-norm penalization
  and optimal rates for noisy low-rank matrix completion,'' \emph{The Annals of
  Statistics}, pp. 2302--2329, 2011.

\bibitem{laurent2000adaptive}
B.~Laurent and P.~Massart, ``Adaptive estimation of a quadratic functional by
  model selection,'' \emph{Annals of Statistics}, pp. 1302--1338, 2000.

\bibitem{Vershynin2012}
R.~Vershynin, ``Introduction to the non-asymptotic analysis of random
  matrices,'' \emph{Compressed Sensing, Theory and Applications}, pp. 210 --
  268, 2012.

\end{thebibliography}

\end{document}